\documentclass[a4paper,10pt]{article}
\usepackage[utf8x]{inputenc}
\usepackage{multicol}
\usepackage{amsthm}
\usepackage{amsmath}
\usepackage{amssymb}
\usepackage{textpos}
\usepackage{bbold}
\usepackage{mathrsfs}
\usepackage{url}
\newcommand{\ketbra}[2]{|{#1}\rangle \langle {#2}|}
\newcommand{\ketbrat}[4]{| \widetilde{#1}_{#2} \rangle \langle \widetilde{#3}_{#4}|}

\newcommand{\ket}[1]{| {#1} \rangle}
\newcommand{\tket}[2]{\ket{{\widetilde{#1}}_{#2}}}
\newcommand{\braket}[2]{\langle #1 | #2 \rangle}
\newcommand{\tbraket}[4]{\langle \widetilde{#1}_{#2} | \widetilde{#3}_{#4} \rangle}

\newcommand{\ens}{\mathcal{E}(r_1,r_2,\cdots,r_m)}
\newcommand{\pro}{\mathcal{P}(r_1,r_2,\cdots,r_m)}
\newtheorem{theorem}{Theorem:}[subsection]
\newtheorem{corollary}[theorem]{Corollary:}
\usepackage{graphicx}
\usepackage{graphics}
\usepackage{graphicx}
\usepackage{tabularx}
\usepackage{tikz}
\usepackage[nodisplayskipstretch]{setspace}
\usepackage{pdflscape}
\usepackage[toc,page]{appendix}
\usepackage{authblk}
\title{Algebraic Structure of the Minimum Error Discrimination Problem for Linearly Independent Density Matrices}
\author[1]{Tanmay Singal\thanks{stanmay@imsc.res.in}}
\author[1]{Sibasish Ghosh\thanks{sibasish@imsc.res.in}}
\affil[1]{Optics and Quantum Information Group, Institute of Mathematical Sciences, CIT Campus, Taramani, Chennai, 600 113, India}

\begin{document}

\maketitle

\begin{abstract}
The minimum error discrimination problem for ensembles of linearly independent pure states are known to have an interesting structure; for such a given ensemble the optimal POVM is given by the pretty good measurment of another ensemble which can be related to the former ensemble by a bijective mapping $\mathscr{R}$ on the ``space of ensembles''. In this paper we generalize this result to ensembles of general linearly independent states (not necessarily pure) and also give an analytic expression for the inverse of the map, i.e., for $\mathscr{R}^{-1}$. In the process of proving this we also simplify the necessary and sufficient conditions that a POVM needs to satisfy to maximize the probability of success for the MED of an LI ensemble of states. This simplification is then employed to arrive at a rotationally invariant necessary and sufficient conditions of optimality. Using these rotationally invariant conditions it is established that every state of a LI mixed state ensemble can be resolved to a pure state 
decomposition so that the corresponding pure state ensemble (corresponding to pure states of all mixed states together) has as its optimal POVM a pure state decomposition of the optimal POVM of mixed state ensemble. This gives the necessary and sufficient conditions for the PGM of a LI ensemble to be its optimal POVM; another generalization for the pure state case. Also, these rotationally invariant conditions suggest a technique to give the optimal POVM for an ensemble of LI states. This technique is polynomial in time and outpeforms standard barrier-type interior point SDP in terms of computational complexity.

\end{abstract}
%\begin{multicols}{2}
\section{Introduction}
\label{intro}

Minimum Error Discrimination (MED) is a state hypothesis testing problem in quantum state discrimination. The setting is as follows: Alice selects a state $\rho_i$ with probability $p_i>0$ from an ensemble of $m$ states $\widetilde{P}=\{ p_i >0, \rho_i \}_{i=1}^{m}$, and sends it to Bob, who is then tasked to find the index $i$ from the set $\{1,2,\cdots, m\}$, by performing measurement on the state he receives. His measurement is a generalized POVM of $m$ elements $E = \{ E_i \}_{i=1}^{m}$, and his strategy for hypothesis testing is based on a one-to-one correspondence between the states $ \rho_i \in \widetilde{P}$ and POVM elements $E_i \in E$ such that he will declare having been given $\rho_j$ when his measurement yields the $j$-th outcome. Since the states $\rho_1, \; \rho_2, \; \cdots, \rho_m$ are not necessarily orthogonal they aren't perfectly distinguishable, i.e., there doesn't exist a measurement such that $Tr \left( \rho_i E_j \right) = \delta_{i,j} Tr \left( \rho_i E_i \right), \; \forall \; 1 \
leq i,j \leq m$ unless $Tr \left( \rho_i \rho_j \right) = \delta_{i,j} Tr \left( \rho_i^2 \right), \; \forall \; 1 \leq i,j \leq m$. That $Tr \left( \rho_i E_j \right) \neq 0$ for some $i \neq j$ implies that there may arise a situation where Alice sends the state $\rho_i$ but Bob's measurement yields the $j$-th outcome which leads him to conclude that Alice gave him $\rho_j$. This is an error. The average probability of such error is given by:

\begin{equation}
\label{Pe}
P_e = \sum_{\substack{i,j=1 \\ i\neq j}}^m p_i Tr( \rho_i \Pi_j)
\end{equation}

where $\{ \Pi_i \}_{i=1}^{m}$ represents an m element POVM with $ \Pi_i \geq 0$ and $ \sum_{i=1}^{m} \Pi_i = \mathbb{1}$. 

The average probability of success is given by:

\begin{equation}
 \label{Ps}
P_s =  \sum_{i=1}^m p_i Tr( \rho_i \Pi_i)
\end{equation}

Both probabilities sum up to 1:

\begin{equation}
\label{Sum1}
P_s + P_e =1
\end{equation}

In MED we are given an ensemble and tasked with finding the maximum value that the $P_s$, as defined in equation \eqref{Ps}, attains over the ``space'' of $m$ element POVMs\footnote{It is appropriate to call the set of $m$ element POVMs a space because it is a convex set, inherently implying that there is a notion of additon and scalar multiplication defined between any two points in the set. The restrictions are on the fact that all linear combinations of elements much be convex combinations. Additionally this space is compact.} and the points in the space of $m$ element POVMs where this maximum value is attained.

\begin{equation}
\label{Pmax}
P_{s}^{\text{max}} = \text{Max} \{ P_s \; | \; \{ \Pi_i \}_{i=1}^{m},  \; \Pi_i \geq 0, \; \sum_{i} \Pi_i = \mathbb{1}\} = 1- P_{e}^{\text{min}}
\end{equation}

Despite the innocuous nature of the problem there have been fairly limited class of ensembles for which the problem has been solved analytically. This includes any ensemble with just two states, i.e., when $m=2$ \cite{Hel}, ensembles of any number where the states are equiprobable and lie on the orbit of a unitary \cite{Ban, Chou}, an ensemble of $3$ qubits \cite{Ha}\footnote{In \cite{Ha} a the general recipe to obtain the optimal POVM for an ensemble of any number of qubits states has been lain down.}, and all pure state ensembles for which the pretty good measurement (PGM) associated with a LI pure state ensemble is its optimal POVM as well \cite{Sasaki}. 

In \cite{Mas} it was shown that there exists a relation between an ensemble $\widetilde{P}$ and another ensemble $\widetilde{Q} = \{ q_i \geq 0, \sigma_i \}_{i=1}^{m}$, with the condition that $ supp \left( q_i \sigma_i \right) \subseteq supp \left( p_i \rho_i \right)$, $\forall \; 1 \leq i \leq m$, such that the optimal POVM for MED of $\widetilde{P}$ is given by the pretty good measurement (PGM) of $\widetilde{Q}$. In the case of linearly independent pure state ensembles (LIP), it is known that $\sigma_i = \rho_i, \; \forall \; 1 \leq i \leq m$, and it is also known that $\widetilde{Q}$ is given as a function of $\widetilde{P}$. This function is invertible and an analytic expression for the inverse of the function is known. This relation between a LI pure state ensemble and its optimal POVM is of significance in finding the optimal POVM \cite{Singal}. It is, hence, desirable to know if such a function exists for other classes of ensembles too. 

In \cite{Carlos} it was shown that such a function isn't definable for linearly dependent pure state ensembles. What about mixed states? From \cite{Yohina} we know that the optimal POVM for an ensemble of LI states is a projective measurement where the rank of the $i$-th projector equals the rank of the $i$-th state in the ensemble. As we will later show, this itself exhibits that $rank \left( p_i \rho_i \right) = rank \left( q_i \sigma_i \right)$, $\forall \; 1 \leq i \leq m$, and, since $ supp \left( q_i \sigma_i \right) \subseteq supp \left( p_i \rho_i \right)$, $\forall \; 1 \leq i \leq m$, this implies that $ supp \left( q_i \sigma_i \right) = supp \left( p_i \rho_i \right)$, $\forall \; 1 \leq i \leq m$. This gives us an indication that the aforementioned function may be definable in the general LI state case, i.e., when the states aren't necessarily pure. 

In this paper we establish that such a function is definable and that it is an invertible function as well. Additionally, we give an analytic expression for the inverse function. In the process we also simplify the necessary and sufficient condition that a POVM has to satisfy to be the optimal POVM for an ensemble of linearly independent states. Also, the necessary and sufficient condition is brought to a rotationally invariant form. This form can be exploited to obtain the optimal POVM for the MED of any LI ensemble. These rotationally invariant conditions tell us that for for each ensemble of LI states, there is a corresponding pure state decomposition such that the ensemble corresponding to this pure state decomposition has an optimal POVM which is itself a pure state decomposition of the optimal POVM for the mixed state ensemble. This fact is used to show when the pretty good measurement of an LI ensemble is its optimal measurment; this is also a generalization of the pure state case. Also, the 
rotationally invariant conditions suggest a recipe to obtain the optimal POVM for a LI ensemble of states. This technique is polynomial in time and simple to use. 

The paper is divided into various sections as follows: section \eqref{OPTC} gives the known optimal conditions for the MED of any general ensemble; section \eqref{Mixed} first introduces what is known so far about MED for LI state ensembles and then goes onto establish the main result of the paper, i.e., that every LI state ensemble can be mapped to another LI state ensemble through an invertible map, such that the PGM of the image of the ensemble under the map is the optimal POVM for the MED of the corresponding pre-image ensemble. Establishing the existence of such a map requires a simplification of the known optimality conditions in the case for LI ensembles which we prove. In the same section we also obtain an analaytic expression for the inverse of this map. In section \eqref{compareMEDP} we compare the problem of MED for general LI mixed ensembles with the problem of MED for LI pure state ensembles which are defined on the same Hilbert space $\mathcal{H}$. It is shown that for every LI mixed state 
ensemble has a pure state decomposition whose optimal POVM is itself a pure state decomposition of the optimal POVM of the mixed state ensemble. Section \eqref{Solution} employs the results developed in section \eqref{Mixed} to give an efficient and simple numerical technique to obtain the optimal POVM for the MED of any LI ensemble. 

\section{The Optimum Conditions}
\label{OPTC}

Alice picks a state $\rho_i$ with probability $p_i$ from the ensemble $\widetilde{P} = \{ p_i, \rho_i \}_{i=1}^{m}$ and hand it to Bob for MED. The states $\rho_1, \rho_2, \cdots, \rho_m$ act on a Hilbert space $\mathcal{H}$ of dimension $n$ and $supp \left(p_1\rho_1 \right)$, $supp \left(p_2 \rho_2 \right)$, $\cdots$, $supp \left(p_m\rho_m \right)$ together span $\mathcal{H}$. Bob's task is the optimization problem given by equation \eqref{Pmax}. This optimization is over the space of of $m$ element POVMs, i.e., the space given by $\left\{ \left\{ \Pi_i \right\}_{i=1}^{m}, \text{ where $\Pi_i \geq 0, \; \forall \; 1 \leq i \leq m,  \; \sum_{i}^{m} \Pi_i = \mathbb{1}$} \right\}$, where $\mathbb{1}$ is the identity operator on $\mathcal{H}$. To every constrained optimization problem (called the primal problem) there is a dual problem which provides a lower bound if primal problem is a constrained minimization or an upper bound if the primal problem is a constrained maximization to the objective function being 
optimized in the primal problem. Under certain conditions these bounds are tight implying that one can obtain the solution for the primal problem from its dual. We then say that there is no duality gap between both problems \cite{Boyd}. 

For MED there is no duality gap and the dual problem can be solved to obtain optimal POVM. This dual problem is given as follows \cite{Yuen}:

\begin{equation}
\label{dual}
\text{Min} \; \text{Tr}(Z) \; \ni \; Z-p_i \rho_i \geq 0, \; \forall\; 1 \leq i \leq m.
\end{equation}

Also the optimal $m$-element POVM will satisfy the complementarity slackness condition:

\begin{equation}
\label{cslack}
(Z- p_i \rho_i)\Pi_i= \Pi_i(Z-p_i \rho_i)=0, \, \forall \, 1\leq i \leq m.
\end{equation}

Now summing over $i$ in equation \eqref{cslack} and using the fact that $ \sum_{i=1}^{m} \Pi_i = \mathbb{1}$ we get:

\begin{equation}
\label{Z}
Z= \sum_{i=1}^{m} p_i \rho_i \Pi_i = \sum_{i}^{m} \Pi_i p_i \rho_i.
\end{equation}

Using equation \eqref{Z} in equation \eqref{cslack}, we get:

\begin{eqnarray}
& \Pi_j ( Z - p_i \rho_i) \Pi_i = \Pi_j ( Z - p_j \rho_j) \Pi_i, &\; \forall \; 1 \leq i,j \leq m  \notag \\
\label{St} 
 \Rightarrow & \Pi_j ( p_j \rho_j - p_i \rho_i ) \Pi_i =0, & \; \forall \; 1 \leq i,j \leq m
\end{eqnarray}

Equation \eqref{St} was derived by Holevo \cite{Hol}, separately, without using the dual optimization problem stated in the problem \eqref{dual}. Equation  \eqref{cslack} and equation \eqref{St} are equivalent to each other. These are necessary but not sufficient conditions. Of the set of $m$ element POVMs which satisfy equation \eqref{cslack} (or equivalently equation \eqref{St}) only a proper subset is optimal. This optimal POVM will satisfy the global maxima conditions given below:

\begin{eqnarray}
 & Z \geq p_i \rho_i \notag  &, \; \forall \; 1 \leq i \leq m,\\
\label{Glb}
\Longrightarrow & \sum_{k=1}^{m} p_k \rho_k \Pi_k - p_i \rho_i \geq 0,& \; \forall \; 1 \leq i \leq m.     
\end{eqnarray}
 
Thus the necessary and sufficient conditions for the $m$-element POVM(s) to maximize $P_s$ are given by equations \eqref{cslack} (or equivalently, equation \eqref{St}) and condition \eqref{Glb}. 

\section{Linearly Independent States}
\label{Mixed}

Let $\mathcal{H}$ be an $n$ dimensional Hilbert space. Consider a set of $m$ $(\leq n)$ LI states in $\mathcal{H}$, denoted by $P =\{ \rho_i \}_{i=1}^{m}$, where $\rho_i \in \mathcal{B}(\mathcal{H}), \; \rho_i \geq 0, \; Tr(\rho_i)=1,\; \forall \; 1 \leq i \leq m$. Let $r_i \equiv \text{rank}(\rho_i), \;\forall \; 1 \leq i \leq m$. Also let $\sum_{i=1}^m r_i = n$. This implies that $\mathcal{H}$ is fully spanned by supports of $\rho_1, \rho_2, \cdots, \rho_m$ and that the supports of $\rho_1, \rho_2, \cdots, \rho_m$ are linearly independent. Let elements within $P$ be indexed in descending order of $r_i$, i.e.,  $r_{i}\geq r_{i+1}, \; \forall \; 1 \leq i \leq m-1$. Consider $T  \in \mathcal{B}(\mathcal{H})$ to be non-singular; construct an ensemble $\widetilde{P}'=\{ p_i', \rho_i'\}_{i=1}^{m}$ by a congruence transformation on elements of $P$ by $T$ in the following manner:
 
\begin{subequations}
\begin{equation}
\label{cojens1}
\rho'_i \equiv \dfrac{ T \rho_i T^{\dag}  }{ Tr(T \rho_i T^{\dag})},   
\end{equation}
\begin{equation}
\label{cojens2} 
p'_i \equiv \frac{Tr( T \rho_i T^{\dag}) }{ \sum_{\substack{j=1}}^{m}   Tr( T \rho_j T^{\dag}) }. 
\end{equation}
\end{subequations}

Note (i) $\widetilde{P}'=\{p'_i >0, \, \rho'_i \}_{i=1}^{m}$ is an ensemble of $m$ linearly independent states (ii) rank$(\rho'_i)=r_i, \; \forall \; 1 \leq i \leq m$. 

Let's denote the transformations in equations \eqref{cojens1} and \eqref{cojens2} concisely by: $\widetilde{P}'= T P  T^{\dag}$. Using this define the following set:
\begin{equation}
\label{ens1}
\mathcal{E}(r_1,r_2,\cdots, r_m) \equiv \; \{ T P T^{\dag} \; | \; T \in \mathcal{B}(\mathcal{H}), \; det(T)\neq 0\}
\end{equation}

$\mathcal{E}(r_1,r_2,\cdots, r_m)$ is the set of LI ensembles where the $i$-th state has rank $r_i$. This is a $2n^2 - \sum_{i=1}^{m} r_{i}^{2} -1$ real parameter space. If $r_{k}=r_{k+1}=\cdots = r_{k+s-1}$, then a single ensemble can be represented by $s!$ elements in $\ens$, all of which are equivalent to each other upto a permutation among the $k \text{-th}, (k+1)\text{-th}, \cdots, (k+s-1)\text{-th}$ states\footnote{Allowing for this multiplicity is  \emph{just} a matter of convenience, i.e., one could adopt more criteria to do away with such multiplicities but that complicates the description of $\ens$ and, for that purpose, such a description is avoided.}. Let us now list what is known so far about the optimal POVMs for MED of LI ensembles. 

For the case of pure state ensembles (LIP), i.e., when $r_i=1,$ $\forall \; i = 1,2,\cdots,m$\footnote{Note that in this case $m=n$.}, it is already well known that the optimal POVM is given by a unique rank-one projective measurement \cite{Ken}. There is a corresponding result for general LI ensembles and that was explicitly proved in \cite{Yohina}, although it could also be inferred from \cite{Mas}. Therein, it was shown that the optimal POVM for MED of a LI ensemble $\widetilde{P}$ of $m$ states with ranks $r_1, r_2, \cdots, r_m$ respectively, i.e., such that $\widetilde{P} \in \ens$, is given by a POVM $\{ \Pi_i \}_{i=1}^{m}$ with the relation $rank(\Pi_i)=r_i, \; \forall \; 1 \leq i \leq m$. Note that the linear independence of the states $\rho_1, \rho_2, \cdots, \rho_m$, is contained in the relation: $\sum_{i=1}^{m}r_i = \; dim\mathcal{H} \; (=n)$ and this relation along with the aforementioned condition, that $rank\left( \Pi_i \right)=r_i$, implies that $\{ \Pi_i \}_{i=1}^{m}$ \emph{has to be} 
a projective measurement, i.e, $\Pi_i \Pi_j = \delta_{ij} \Pi_i, \; \forall \; 1 \leq i,j \leq m$. The relation $rank\left( \Pi_i \right) = r_i$ also ensures that the optimal POVM is unique. To establish this consider a case where we know that two $m$-element POVMs are optimal for the MED of some LI ensemble in $\ens$; let these optimal POVMs (which are projective measurments) be denoted by $\{ \Pi_i^{ \left( 1 \right)} \}_{i=1}^{m}$ and $\{ \Pi_i^{ \left( 2 \right)} \}_{i=1}^{m}$. The rank condition tells us that $rank \left( \Pi_i^{ \left( 1 \right) }  \right)= rank \left( \Pi_i^{ \left( 2 \right) } \right)  = r_i$, $\forall \; 1 \leq i \leq m$. The only way that a convex combination of both POVMs of the form $\{ p \Pi_i^{ \left( 1 \right) } + (1-p) \Pi_i^{ \left( 2 \right) } \}_{i=1}^{m}$ ( where $ 0 < p < 1$)\footnote{We need to ensure that the POVM, which is a convex combination, is also an $m$ element POVM. That is why convex combinations are only taken in this form.} also satisfies the rank condition (
that $rank \left( p \Pi_i^{ \left( 1 \right) } + (1-p) \Pi_i^{ \left( 2 \right) } \right) = r_i$, $\forall \; 1 \leq i \leq m$) is if $ \Pi_i^{ \left( 1 \right) } =  \Pi_i^{ \left( 2 \right) }$, $\forall \; 1 \leq i \leq m$. Another way of saying the same thing is that for $0 < p <1$, $\{ p \Pi_i^{\left( 1 \right) } + (1-p) \Pi_i^{ \left( 2 \right) } \}_{i=1}^{m}$ is a projective measurement iff $ \Pi_i^{ \left( 1 \right) } =  \Pi_i^{ \left( 2 \right) }$, $\forall \; 1 \leq i \leq m$.This implies that for MED of any LI ensemble, the optimal POVM is unique. 

We now define a set, which we denote by $\mathcal{P}(r_1,r_2,\cdots,r_m)$. An element $\{ \Pi_i \}_{i=1}^{m} \in \pro$ has the properties: (i) $\sum_{i=1}^{m} \Pi_i = \mathbb{1}$ (ii) $Rank(\Pi_i)=r_i,\; \forall \; 1 \leq i \leq m$  (iii) $\Pi_i \Pi_j = \delta_{ij} \Pi_i$. As noted before, (i) and (ii), along with the relation $\sum_{i=1}^{m} r_i = dim\mathcal{H}$, imply (iii) to hold true. Thus $\mathcal{P}\left(r_1,r_2,\cdots,r_m \right)$ is a subset of the set of projective measurements on $\mathcal{H}$. $\pro$ is an $n^2 - \sum_{i=1}^{m} r_i^{2}$ real parameter set.

The uniqueness of the optimal POVM for MED of an ensemble of LI states implies that one can unambiguously define ``the optimal POVM map" from $\ens$ to $\pro$. Let the optimal POVM map be denoted by $\mathscr{P}$. Then $\mathscr{P}: \ens \longrightarrow \pro$ is such that $\mathscr{P}(\widetilde{P})$ is the unique optimal POVM in $\pro$ for the MED of any ensemble $\widetilde{P} \in \ens$.

In \cite{Mas} it was shown that the optimal POVM for MED of a LI ensemble $\widetilde{P} = \{p_i , \rho_i\}_{i=1}^{m} \in \ens$, i.e., $\mathscr{P} \left( \widetilde{P} \right) $, is the PGM of another ensemble of states $\widetilde{Q} = \{q_i, \sigma_i \}_{i=1}^{m}$, where (1) $q_i \geq 0$, $\sum_{i=1}^{m}q_i=1$ and (2) $supp \left( \sigma_i \right) \subseteq supp \left( \rho_i \right)$, for all $ 1 \leq i \leq m$. If we denote $\mathscr{P} \left( \widetilde{P} \right) $ as $ \{ \Pi_i \}_{i=1}^m$, then $\Pi_i$ has the form\footnote{Note that $\left( \sum_{j=1}^{m} q_j \sigma_j \right)^{-\frac{1}{2}}$ is well defined because $\sum_{j=1}^{m} q_j \sigma_j > 0$. This is the consequence of the fact that the supports of $\rho_1, \rho_2, \cdots, \rho_m$ span $\mathcal{H}$}:

\begin{equation}
\label{formPI}
\Pi_i = \left( \sum_{j=1}^{m} q_j \sigma_j \right)^{-\frac{1}{2}} \; q_i \sigma_i \; \left( \sum_{k=1}^{m} q_k \sigma_k \right)^{-\frac{1}{2}}.
\end{equation}
 
In the LIP case, i.e., when $r_i =1, \; \forall \; 1 \leq i \leq m$, we know the following:
 
\begin{enumerate}

\item{ $q_i > 0, \; \forall \; 1 \leq i \leq m$}

\item{$supp(\rho_i) = supp(\sigma_i), \; \forall \; 1 \leq i \leq m$\footnote{Since in the LIP case, $\rho_i$ are all rank one, this means $\rho_i = \sigma_i , \; \forall \; 1 \leq i \leq m$}}

\item {The correspondence $\widetilde{P} \rightarrow \widetilde{Q}$ is a map, and it is an invertible map. An analytic expression for the inverse map, i.e. the map from $\widetilde{Q} \rightarrow \widetilde{P}$, was obtained in \cite{Bela, Mas, Carlos}.}
\end{enumerate}

We are motivated to answer the question whether these results can be extended to cases where $r_i \geq 1$? We already noted that $rank \left(\Pi_i \right)=rank \left( \rho_i \right)$, $ \forall \; 1 \leq i \leq m$. This implies that (1) $q_i >0$\footnote{Had $q_i=0$ for any $i=1,2,\cdots, m$, $ \Pi_i =0$ (see equation \eqref{formPI}). We know that this isn't true because $rank(\Pi_i)=r_i \neq 0$.} and (2) $supp \left( \sigma_i \right) $ $ =$ $ supp \left( \rho_i \right)$\footnote{Since $\sigma_i$ and $\Pi_i$ are related through a congruence transformation $\forall \; 1 \leq i \leq m$ (see equation \eqref{formPI}) it follows that $rank(\sigma_i) = rank \left( \Pi_i \right) = r_i$. Since $supp(\sigma_i)$ is a subspace of $supp(\rho_i)$ and since $rank(\rho_i)= r_i = rank(\sigma_i)$ it follows that $supp(\sigma_i) = supp(\rho_i)$, $\forall \; 1 \leq i \leq m$.}. In this paper we establish that (3) holds for general LI ensembles too, i.e., we first establish that the correspondence $\widetilde{P} \rightarrow \
widetilde{Q}$ is a mapping, then we prove that this is an invertible map and we give an analytic expression for the inverse of this map. Later on we will use the existence of this map to derive a technique to obtain the optimal POVM for a LI ensemble, in the same way as done for LI pure state ensembles in \cite{Singal}.

For this purpose defnie the PGM map from $\ens$ to $\pro$ such that $PGM \left( \widetilde{Q} \right)$ is the pretty good measurment associated with the ensemble and the PGM of the ensemble $\widetilde{Q} = \{q_i, \sigma_i\}_{i=1}^{m}$ is defined by:

\begin{equation}
\label{PGM}
PGM \left( \widetilde{Q} \right) = \{ \left(\sum_{j=1}^{m} q_j \sigma_j \right)^{-\frac{1}{2}} q_i \sigma_i \left(\sum_{k=1}^{m} q_k \sigma_k \right)^{-\frac{1}{2}} \}_{i=1}^{m}.
\end{equation}

\subsection{The $\widetilde{P} \rightarrow \widetilde{Q}$ Correspondence:}
\label{PQcorr}

Given that $\mathscr{P} \left( \widetilde{P} \right) = \{\Pi_i \}_{i=1}^{m}$, where $\{\Pi_i \}_{i=1}^{m} \in \pro$. Hence $ \Pi_i \Pi_j = \delta_{ij} \, \Pi_i, \; \forall \; 1 \leq \, i,j \, \leq m$. Consider a spectral decomposition of each $\Pi_i$ into pure states:
\begin{equation}
\label{Pidecomposition}
\Pi_i = \sum_{j=1}^{r_i} \ketbra{w_{ij}}{w_{ij}},
\end{equation}
where $\braket{w_{i_1j_1}}{w_{i_2j_2}}=\delta_{i_1i_2}\delta_{j_1j_2}$ for $ 1 \leq i_1, i_2 \leq m$ and $1 \leq j_1 \leq r_{i_1}$, $1 \leq j_2 \leq r_{i_2}$. For each $\Pi_i$ there is a $U\left(r_i\right)$ degree of freedom in choosing this spectral decomposition. For now we assume that $\{ \ketbra{w_{ij}}{w_{ij}} \}_{j=1}^{r_i}$ is any spectral decomposition of $\Pi_i$ in equation \eqref{Pidecomposition}. Later on a specific choice of the set $\{ \ket{w_{ij}} \}_{i=1, j=1}^{i=m,j=r_i}$ will be made. 

Each of the unnormalized density matrices $p_i \rho_i$ can be decomposed into a sum of $r_i$ pure states in the following way:

\begin{equation}
\label{rhodecomposition}
p_i \rho_i = \sum_{j_i=1}^{r_i} \ketbrat{\psi}{ij_i}{\psi}{ij_i}.
\end{equation} Here the vectors $\tket{\psi}{ij_i}$ are unnormalized. And the set $\{ \tket{\psi}{ij_i} \}_{j_i=1}^{r_i}$ is LI. Again there is a $U\left( r_i \right)$ degree of freedom in the choice of decomposition of the unnormalized state $p_i \rho_i$ into the vectors $\tket{\psi}{ij_i}$. We assume that some choice of such a decomposition has been made in equation \eqref{rhodecomposition} without any particular bias. Let the gram matrix corresponding to the set $\{ \tket{\psi}{ij_i}  \; | \; 1 \leq i \leq m,  \; 1 \leq j_i \leq r_i\}$ be denoted by $G$, whose matrix elements are given by the following equation: 
\begin{equation}
\label{Gram2}
G^{(l \; i)}_{k_l \; j_i}= \tbraket{\psi}{l k_l}{\psi}{i j_i} 
\end{equation} Some explanation on the indices is in order. All the $n \times n$ matrices that we deal with in this paper are divided into blocks of sizes $r_1, r_2, \cdots, r_m$. The matrix element of such an $n \times n$ matrix is given by two tiers of row indices and two tiers of column indices: the inter-block  row (or column) index and the intra-block row (or column) index. The former are represented by the superscript $(l \; i)$, where $l$ represents the row block and $i$ represents the column block in the $n \times n$ matrix, whereas the latter are represented by subscripts $k_l \; j_i$, where $k_l$ represents the $k$-th row and $j_i$ the $j$-th column of the $(l \; i)$-th matrix block of the $n \times n$ matrix. This implies that $1 \leq k_l \leq r_l$ and $1 \leq j_i \leq r_i$. At times subscripts $l$ in $k_l$ and $i$ in $j_i$ are omitted. In such situations it is clear which block the intrablock indices $k$ and $j$ are for. This notation, while at first seems cumbersome,  will come in handy later. 

For each $i=1,2,\cdots, m$, the set $\{ \tket{\psi}{ij_i} \}_{j_i=1}^{r_i}$ is LI. Since $supp \left(p_1 \rho_1 \right),$ $supp \left(p_2 \rho_2 \right),$ $\cdots$, $supp \left(p_m \rho_m \right)$ are LI, the set $\bigcup_{i=1}^{m} \{ \tket{\psi}{ij_i} \}_{j_i=1}^{r_i} $  is LI as well. This implies that $G>0$. Corresponding to the set $\bigcup_{i=1}^{m} \{ \tket{\psi}{ij_i} \}_{j_i=1}^{r_i} $ there is another set of vectors given by: $\{ \tket{u}{ij_i} \}_{i=1, j_i = 1}^{i=m, j_i = r_i}$ with the property:

\begin{equation}
\label{psiandu}
\tbraket{\psi}{i_1j_1}{u}{i_2j_2} = \delta_{i_1 i_2} \delta_{j_1 j_2}, \; \forall \; 1 \leq i_1,i_2 \leq m \text{ and } 1 \leq j_1 \leq r_{i_1}, \; 1 \leq j_2 \leq r_{i_2}. 
\end{equation} 

The vectors $\tket{u}{ij_i}$ can be expanded in the basis $\{ \tket{\psi}{ij} \}_{i=1,j_i=1}^{i=m,j_i=r_i}$ in the following way:

\begin{equation}
\label{u}
\tket{u}{ij_i} = \sum_{l=1}^{m}\sum_{l_k=1}^{r_l} \left( G^{-1} \right)^{(l \; i)}_{k_l \; j_i} \tket{\psi}{lk_l}, \; \forall \; 1 \leq i \leq m, \; 1 \leq j_i \leq m.
\end{equation}

From equation \eqref{u} it can be seen that the set $\{ \tket{u}{ij_i} \}_{i=1, \; j_i=1}^{i=m, \; j_i=r_i}$ is a LI set of $n$ vectors. Hence it forms a basis for $\mathcal{H}$. This is also corroborated by the fact that the gram matrix of the set $\{ \tket{u}{ij_i} \}_{i=1, \; j_i=1}^{i=m, \; j_i=r_i}$ is $G^{-1}$. Thus the orthonormal basis vectors $\{ \ket{ w_{ij_i} } \}_{i=1, \; j_i=1}^{i=m, \; j_i=r_i}$, given by equation \eqref{Pidecomposition}, can be expanded in terms of the $\tket{u}{ij_i}$ vectors:

\begin{equation} 
\label{wexpandu}
\ket{w_{ij_i}} = \sum_{l=1}^{m}\sum_{k_l=1}^{r_l} \left( G^\frac{1}{2} W \right)^{(l \; i)}_{k_l \; j_i} \tket{u}{lk_l}, \; \forall \; 1 \leq i \leq m, \; 1 \leq j_i \leq r_i
\end{equation} where $W$ is an $n \times n$ unitary matrix. There is a one-to-one correspondence between the unitary matrix $W$ and the choice of spectral decomposition in equation \eqref{Pidecomposition}, i.e., fixing the spectral decomposition of the projectors $\Pi_i$ in equation \eqref{Pidecomposition} fixes the unitary $W$ uniquely. This becomes clearer in the following equation:
 
Substituting equation \eqref{wexpandu} in equation \eqref{Pidecomposition} we get:
\begin{equation}
\label{mixedP}
\Pi_i = \sum_{l_1, l_2 = 1}^{m} \sum_{k_1=1}^{r_1} \sum_{k_2=1}^{r_2} \left( \sum_{j=1}^{r_i} \left(  G^{\frac{1}{2}}W \right)^{(l_1 \; i)}_{k_1  \;  j}  \left( W^\dag G^{\frac{1}{2}} \right)^{(i  \;  l_2)}_{j  \; k_2} \right) \ketbrat{u}{l_1 k_1}{u}{l_2 k_2}.
\end{equation}

%The LHS remains of equation \eqref{mixedP} remains invariant if, in the RHS, $W$ undergoes a transformation of the form $W \longrightarrow \left( \bigoplus_{i=1}^{m}U^{\left( i \right) } \right) W \left( \bigoplus_{i=1}^{m} {U^{\left( i \right) } }^\dag \right)$, where $U^{\left( i \right) } \in U \left( r_i \right),$ $\forall \; 1 \leq i \leq m$. Under the transformation of $W$, the eigenbasis $\{ w_{i j_i} \}_{i=1, \; j_i = 1}^{i=m, \; j_i=r_i}$ given by equation \eqref{wexpandu} transforms accordingly, which implies that the spectral decomposition of $\Pi_i$ in equation \eqref{Pidecomposition} transforms accordingly as well.

Upon substituting the expression for $\Pi_i$ and $\Pi_j$ from equation \eqref{mixedP} into equation \eqref{St} we get the following:
\begin{align}
\label{St3}
\sum_{\substack{ 1 \leq l_1, \; l_2 \leq m, \\  1\leq k_1 \leq r_{l_1}, \\ 1\leq k_2 \leq r_{l_2}   }}
 \xi^{(l_1 \; l_2)}_{k_1  \; k_2}  \ketbrat{u}{l_1 k_1}{u}{l_2 k_2} \; = \; 0
\end{align} where $ \xi^{(l_1 l_2)}_{k_1 k_2}$ is given by:
\begin{align}
\label{xi}
&  \xi^{(l_1  \;  l_2)}_{k_1  \;  k_2} = \notag \\
&\sum_{s=1}^{r_i}\sum_{t=1}^{r_j}
 \left(G^{\frac{1}{2}}W\right)^{\left( l_1  \;  i \right)}_{k_1  \;  s} 
 \left(
 \sum_{h=1}^{r_i}
 \left( W^\dag G^{\frac{1}{2}} \right)^{\left(i  \;  i \right)}_{s  \;   h}
 \left( G^{\frac{1}{2}}W \right)^{\left(i   \;  j \right)}_{h  \;  t}
 -
 \sum_{g=1}^{r_j}
 \left( W^\dag G^{\frac{1}{2}} \right)^{\left(i  \;  j \right)}_{s  \;  g}
 \left( G^{\frac{1}{2}}W \right)^{\left(j  \;  j  \right)}_{g  \;  t}
 \right)
 \left( W^\dag G^{\frac{1}{2}}\right)^{\left(j  \;  l_2 \right)}_{t  \;  k_2} 
\end{align}
Equation \eqref{St3} is the stationary condition \eqref{St}. The expression for $\xi^{\left( l_1 \; l_2 \right)}_{k_1 \; k_2}$ in equation \eqref{xi} is pretty complicated. It is desired make equation \eqref{St3} more transparent. With this aim in mind we partition the matrix $G^{\frac{1}{2}} W$ into the aforementioned blocks and introduce a notation for these blocks:
\begin{enumerate}
\item \begin{equation}
\label{part1}
G^{\frac{1}{2}} W   = \begin{pmatrix}
                       X^{(11)} & X^{(12)} & \cdots & X^{(1m)}\\
                       X^{(21)} & X^{(22)} & \cdots & X^{(2m)}\\
                       \vdots & \vdots & \ddots & \vdots\\
                       X^{(m1)} & X^{(m2)} & \cdots & X^{(mm)}
                       \end{pmatrix}
\end{equation} where $X^{(l_1l_2)}$ is the $\left( l_1 l_2 \right) $-th block of dimension $r_{l_1} \times r_{l_2}$ in $G^\frac{1}{2}W$. The matrix elements of $X^{(l_1 l_2)}$ are given by $ \left( X^{(l_1 l_2)} \right)_{k_1 \; k_2} = \left( G^{\frac{1}{2} } W \right)^{\left( l_1 \; l_2 \right)}_{k_1 \; k_2}, \; \forall \, 1 \leq l_1, l_2 \leq m,$ $ \forall \, 1 \leq k_1 \leq r_{l_1}, \; 1 \leq k_2 \leq r_{l_2}$. 
\item Define: \begin{align}
 \label{column}
& C^{(i)} \equiv \begin{pmatrix}
        X^{(1i)}\\
        X^{(2i)}\\
        \vdots \\
        X^{(mi)} 
       \end{pmatrix},  \; 1 \leq i \leq m 
 \end{align} Thus $C^{(i)}$ is the $i$-th block column of $G^{\frac{1}{2}}W$.
\item Similarly, let's partition $W^\dag G^{-frac{1}{2}}$ into blocks: 
\begin{equation}
\label{part2}
W^{\dag}G^{-\frac{1}{2}} \, = \begin{pmatrix}
                                Y^{(11)} & Y^{(12)} & \cdots & Y^{(1m)}\\
                                Y^{(21)} & Y^{(22)} & \cdots & Y^{(2m)}\\
                                \vdots & \vdots & \ddots & \vdots\\
                                Y^{(m1)} & Y^{(m2)} & \cdots & Y^{(mm)}
                               \end{pmatrix} 
\end{equation} where $\left( Y^{(l_1 l_2)}\right)_{k_1 k_2} = \left( W^{\dag} G^{-\frac{1}{2}} \right)^{\left( l_1 \; l_2 \right)}_{k_1 \; k_2}, \; \forall \, 1 \leq l_1, l_2 \leq m, \; 1 \leq k_1 \leq r_{l_1}$ and $1 \leq k_2 \leq r_{l_2}$. 
\item Define:
\begin{align}
\label{row}
& R^{(i)} \equiv \begin{pmatrix}
       Y^{(i1)} & Y^{(i2)} & \cdots & Y^{(im)}
       \end{pmatrix}, \; 1 \leq i \leq m 
\end{align} Thus $R^{(i)}$ is the $i$-th block-row of $W^{\dag} G^{-\frac{1}{2}}$. 
\end{enumerate}
Substituting equations \eqref{part1} and \eqref{column} in equation \eqref{St3} we obtain condition \eqref{St} in a more transparent form:
\begin{equation}
\label{St5}
C^{(i)} \left(  {X^{(ii)}}^{\dag} X^{(ij)} - {X^{(ji)}}^{\dag} X^{(jj)} \right) {C^{(j)}}^{\dag}\, = \,0, \quad \forall \, 1 \leq i,j \leq m
\end{equation} where $ {X^{ji}}^{\dag}$ is the $(ij)$-th block of $W^\dag G^{\frac{1}{2}}$.
From the definition of equations \eqref{column} and \eqref{row}, $R^{(i)} C^{(i)} = \mathbb{1}_{r_i}, \quad \forall \, 1 \leq i \leq m$ where $\mathbb{1}_{r_i}$ is the identity matrix of dimension $r_i$. Left and right multiplying the LHS and RHS of equation \eqref{St5} by $R^{(i)}$ and ${R^{(j)}}^\dag$ respectively gives: 
\begin{eqnarray}
& R^{(i)} \, C^{(i)} \left(  {X^{(ii)}}^{\dag} X^{(ij)} - {X^{(ji)}}^{\dag} X^{(jj)} \right) {C^{(j)}}^{\dag} \, {R^{(j)}}^{\dag} & = \,0 \notag \\
\label{St5'}
\Longrightarrow & {X^{(ii)}}^{\dag} X^{(ij)} - {X^{(ji)}}^{\dag} X^{(jj)} & =0, \; \forall \; 1 \leq i,j \leq m.
\end{eqnarray} Let $U_D$ be a block diagonal unitary matrix given in the following equation:
\begin{equation}
\label{UD}
U_D = \begin{pmatrix}
U^{(1)} & 0 & \cdots & 0 \\
0 & U^{(2)} & \cdots & 0 \\
\vdots & \vdots & \ddots & \vdots\\
0 & 0 & \cdots & U^{(m)}
\end{pmatrix} 
\end{equation} where $U^{(i)}$ is an $r_i \times r_i$ unitary matrix for $i=1,2,\cdots,m$. We remarked earlier that there is a $U(r_i)$ degree of freedom in choice of resolution of spectral decomposition of $\Pi_i$ in equation \eqref{Pidecomposition}. What that means is that $\Pi_i$ is invariant under the transformation: $\ket{w_{ij}} \rightarrow \ket{w_{ij}'} = \sum_{k=1}^{r_i} U^{(i)}_{k j} \ket{w_{ik}} =\sum_{l=1}^{m}\sum_{k=1}^{r_l} \left(U_D \right)^{\left( l \; i \right)}_{k \; j} \ket{w_{lk}} $, where $1 \leq i \leq m, 1 \leq j \leq r_i$. Expanding the vectors $\ket{w'_{ij}}$ in the basis $\{ \tket{u}{ij_i} \}_{i=1, j_i=1}^{r=m,j_i=r_i}$ gives:

\begin{equation}
\label{wij'}
\ket{w_{ij}'} = \sum_{l=1}^{m}\sum_{k=1}^{r_l} \left( G^{\frac{1}{2}}WU_D \right)^{(l \; i)}_{k \; j}  \tket{u}{lk}.
\end{equation}

It is readily seen that this will leave $\Pi_i$ invariant in equation \eqref{mixedP}. Here we make a specific choice of $U_D$, which is so that the diagonal blocks of $G^{\frac{1}{2}}WU_D$ are positive semidefinite, i.e., $X^{(ii)}U^{(i)} \geq 0, \; \forall \, 1 \leq i \leq m $\footnote{ Given some arbitrary choice of spectral decomposition for $ \Pi_1, \Pi_2, \cdots, \Pi_m $ and the fixed unitary $W$ that corresponds to these spectral decompositions, choose $ U^{(1)},U^{(2)},\cdots,U^{(m)}$ such that $X^{(ii)}U^{(i)} \geq 0, \quad \forall \, 1 \leq i \leq m $. It is always possible to find some $U^{(i)}$ such that $X^{(ii)}U^{(i)} \geq 0$ using singular value decomposition of $X^{(ii)}$. Moreover once the non-singularity of the $X^{(ii)}$ matrices has been established (proved in theorem \eqref{Xii}), the unitaries $ U^{(1)},U^{(2)},\cdots,U^{(m)}$ are unique for a given the spectral decompositions of $\Pi_i$'s (and the associated $W$).}. From here onwards we assume that $U_D$ is absorbed within $W$, i.e., 
$WU_D \rightarrow W$, $X^{(ij)} U^{(j)} \rightarrow X^{(ij)}$ and $\ket{w'_{i j_i}} \longrightarrow \ket{w_{i j_i}}$. This establishes that for any given decomposition of the unnormalized states $p_i \rho_i$ into pure unnormalized states $\tket{\psi}{i j_i}$, as in equation \eqref{rhodecomposition}, there is a unique unitary $W$ such that (1) the ONB $\{ \ket{w_{i j_i}} \}_{i=1, \; j_i=1}^{i=m, \; j_i=r_i}$, defined by equation \eqref{wexpandu}, corresponds to the optimal POVM, in equation \eqref{Pidecomposition} and (2) the matrix $G^\frac{1}{2}W$, which occurs in the equation \eqref{wexpandu}, has positive semi-definite block diagonal matrices\footnote{As mentioned in the footnote above, it is only when we prove that $X^{(ii)}$'s are non-singular, that it will be clear that there exists a unique $U^{(i)}$ such that $X^{(ii)}U^{(i)}>0$. And only then will it be clear that $W \longrightarrow W U_D$ is unique. As it stands now, the non-singularity of the $X^{(ii)}$'s still remains to be proved.} ( i.e., $X^{\left( i i \right)} \geq 0, \; \forall \
; 1 \leq i \leq m$). This point should be kept in mind since it will be crucial later. Thus equation \eqref{St5'} becomes:
\begin{equation}
\label{St6}
X^{(ii)} X^{(ij)} - {X^{(ji)}}^{\dag} X^{(jj)}  =0, \; \forall \; 1 \leq i,j \leq m
\end{equation} 
Define $D$ as the block diagonal matrix containing diagonal blocks of $G^{\frac{1}{2}}W$:
\begin{equation}
\label{DX}
D \equiv \begin{pmatrix}
            {X^{(11)}} & 0 &  \cdots & 0 \\
            0 & {X^{(22)}} &  \cdots & 0\\
            \vdots & \vdots & \ddots & \vdots \\
            0 & 0 & \cdots & {X^{(mm)}}
           \end{pmatrix}
\end{equation}
Left multiplying $G^{\frac{1}{2}}W$ by $D$ gives:
\begin{equation}
\label{DXG}
D G^{\frac{1}{2}}W = \begin{pmatrix}
                                    (X^{(11)})^2 &  {X^{(11)}} X^{(12)} & \cdots &  {X^{(11)}} X^{(1m)} \\
                                    {X^{(22)}} X^{(21)} & (X^{(22)})^2 & \cdots &   {X^{(22)}} X^{(2m)}\\
                                    \vdots & \vdots & \ddots & \vdots \\
                                    {X^{(mm)}} X^{(m1)} & {X^{(mm)}} X^{(m2)} & \cdots & (X^{(mm)})^2
                                   \end{pmatrix}
\end{equation} 
Equation \eqref{St6} tells us that $D G^{\frac{1}{2}}W$ is a hermitian matrix. From that we get:
\begin{equation}
\label{Ainv}
 \left( DG^\frac{1}{2}W  \right)^2 = \left( DG^\frac{1}{2}W  \right) \; \left( W^\dag G^\frac{1}{2} D \right)  = DGD
\end{equation}

Thus condition \eqref{St} implies that one needs to find a block diagonal matrix, $D=$ $ Diag ($ $X^{(11)},$ $X^{(22)}$ $,\cdots,$ $X^{(mm)} ) $ $\geq0$ where $X^{(ii)}$ is an $r_i \times r_i$ positive semidefinite matrix, so that the diagonal blocks of \emph{a} hermitian square root of the matrix $DGD$ are given by $\left( X^{(11)} \right)^2, \left( X^{(22)} \right)^2, \cdots,\left( X^{(mm)} \right)^2$ respectively. Here $G$ corresponds to the gram matrix of vectors $\{ \tket{\psi}{ij_i} \; | \; 1 \leq i \leq m, \; 1 \leq j_i \leq r_i\}$ where $p_i\rho_i = \sum_{j_i=1}^{r_i} \ketbrat{\psi}{ij_i}{\psi}{ij_i}$, for all $i=1,2,\cdots,m$. This is a rotationally invariant condition, i.e., these optimality conditions enable us to get the optimal POVM for any ensemble of the form $U \widetilde{P} U^\dag$, where $U \in U(n)$. 

Condition \eqref{St} is only one of the necessary and sufficient conditions that the optimal POVM needs to satisfy. The other condition is given by condition \eqref{Glb}. We will prove that both conditions can be subsumed in the statement that $D G^\frac{1}{2} W >0$. We can already see that condition \eqref{St} is contained in the statement $DG^\frac{1}{2}W>0$ because positivity of a matrix subsumes hermiticity as well. But to establish the positivity we first need to prove that $D G^\frac{1}{2} W$ is non-singular for which we only need to establish that $D$ is non-singular (since $G^\frac{1}{2} >0$ and $W$ is unitary, $G^\frac{1}{2} W$ is non-singular). To prove that $D$ is non-singular is equivalent to proving that $X^{(ii)}$ are non-singular, i.e., $X^{(ii)}$ is of rank $r_i$ for all $ 1 \leq i \leq m$.

\begin{theorem}
\label{Xii}
$X^{(ii)}$ is of rank $r_i$, $\forall \; 1 \leq i \leq m$.
\end{theorem}
\begin{proof}
Using equations \eqref{rhodecomposition}, \eqref{mixedP} and \eqref{part1}, the operator $p_i^2 \rho_i \Pi_i \rho_i$ can be expanded in the following operator basis, $\{ \ketbrat{\psi}{ij}{\psi}{lk}\; | \; 1 \leq i, \; l \leq m; \; 1 \leq j \leq r_i, \; 1 \leq k \leq r_l \}$. This gives:
\begin{align}
\label{ka}
  p_i^2 \rho_i \Pi_i \rho_i  =  \sum_{\substack{j, \; k=1}}^{r_i}  \left( { \left( X^{\left( ii \right) }\right)}^2 \right)_{jk} \ketbrat{\psi}{ij}{\psi}{ik}.
\end{align}
Now we know that $rank \left( \Pi_i \right) = r_i, \; \forall \; 1 \leq i \leq m$. So $rank \left( p_i^2 \rho_i \Pi_i \rho_i \right), rank \left( p_i \rho_i \Pi_i \right) ( = rank\left( p_i \Pi_i \rho_i  \right) ) $ $\leq r_i,  \; \forall \; 1 \leq i \leq m$. We first establish that $rank \left( p_i \rho_i \Pi_i \right) =  rank \left( p_i \Pi_i \rho_i \right) = r_i$. Suppose not, i.e., let $rank \left( p_k \rho_k \Pi_k \right) < r_i$. This implies that $\exists \; \ket{v} \in supp \left( \Pi_k \right)- \{0\}$ $\ni \; p_k \rho_k \Pi_k \ket{v} =0$. But since $\Pi_j \ket{v} = 0$ when $j \neq k$ \footnote{ $\ket{v} \in supp{ \left( \Pi_i \right)}$ and $\Pi_i \Pi_j = \Pi_i \delta_{ij}, \; \forall \; 1 \leq i,j \leq m$ implies that $\ket{v} \notin supp \left( \Pi_j \right)$.}, we get that $Z \ket{v} = \sum_{i=1}^{m} p_i \rho_i \Pi_i \ket{v} = 0$ using equation \eqref{Z}. This in turn implies that $Z$ cannot be non-singular. But the optimality condition \eqref{Glb} demands that $Z > 0$. Hence the assumption that 
$rank \left( p_i \rho_i \Pi_i \right) < r_i$ isn't true for any $1 \leq i \leq m$. This implies that  $rank \left( p_i \rho_i \Pi_i \right) = r_i$, $\forall \; 1 \leq i \leq m$.
 
That $rank \left( p_i \rho_i \Pi_i \right)$ $ = rank \left( p_i \Pi_i \rho_i  \right)$ $ = rank \left( \Pi_i \right)$ $ = rank \left( \rho_i \right)$ $ = r_i$ implies that any non-zero vector belonging to $supp \left(\Pi_i \right)$ has a non-zero component in $supp \left( \rho_i \right)$ and vice versa for all $1 \leq i \leq m$. 

This tells us that $\rho_i \ket{v} \neq 0 \Rightarrow p_i^2 \rho_i \Pi_i \rho_i \ket{v} \neq 0, \; \forall \; \ket{v} \in \mathcal{H}$, i.e., $supp \left( \rho_i \right) \subseteq supp \left( p_i^2 \rho_i \Pi_i \rho_i \right)$. We already know that  $supp \left( p_i^2 \rho_i \Pi_i \rho_i \right) \subseteq supp\left( \rho_i \right)$. This implies $supp \left( p_i^2 \rho_i \Pi_i \rho_i \right) = supp \left( \rho_i \right)$ which, in turn, implies that $rank \left( p_i^2 \rho_i \Pi_i \rho_i \right) = r_i$, $\forall \; 1 \leq i \leq m$. Using equation \eqref{ka}, this implies that $\left(X^{(ii)}\right)^2$ is of rank $r_i$ and that implies that $X^{(ii)}$ is of rank $r_i$ for all $1 \leq i \leq m$.   
\end{proof}
Theorem \eqref{Xii} implies that $D>0$. And this in turn implies that $D G^\frac{1}{2} W$ is non-singular. We want to now show that the necessary and sufficient optimality conditions given by equation \eqref{cslack} (or equivalently, \eqref{St}) and the inequality \eqref{Glb} are equivalent to the statement that $D G^\frac{1}{2} W >0 $, where $DG^\frac{1}{2}W$ is the matrix occuring in equation \eqref{DXG}. To show that we first need to simplify the optimal POVM conditions for linearly independent states. Let us define a new set of vectors $\{ \tket{\chi}{ij_i} \; | \; 1 \leq i \leq m, \; 1 \leq j_i \leq r_i \}$.

\begin{equation}
\label{chi}
\tket{\chi}{ij_i} = \sum_{\substack{k_i=1}}^{r_i}  X^{(ii)}_{k_ij_i} \tket{\psi}{ik_i}, \; \forall \; 1 \leq i \leq m, \; 1 \leq j_i \leq r_i.
\end{equation}

Since $rank \left( X^{(ii)} \right) = r_i$,  $\left\{ \tket{\chi}{ij} \right\}_{j=1}^{r_i}$ is a basis for $Supp(p_i\rho_i)$. And $\{ \tket{\chi}{ij_i} \}_{i=1, j_i=1}^{i=m, j_i = r_i}$ is a basis for $\mathcal{H}$.

Now the inner product of any two vectors from the set $\{ \tket{\chi}{i j_i} \}_{i=1, \; j_i = 1}^{i=m, \; j_i=r_i}$ is given by:

\begin{equation}
\label{innerchi}
\tbraket{\chi}{i_1 j_1}{\chi}{i_2 j_2} = \left( DGD \right)^{ \left( i_1 \; i_2 \right) }_{j_1 \; j_2}, \; \forall \; 1 \leq i_1, i_2 \leq m, \; 1 \leq j_1 \leq r_{i_1}, \; 1 \leq j_2 \leq r_{i_2}
\end{equation}

This shows us that the gram matrix of the set of vectors $\{ \tket{\chi}{i j_i} \}_{i=1, \; j_i = 1}^{i=1, \; j_i=1r_i}$ is the matrix $DGD$. 

Using this basis we simplify the necessary and sufficient conditions for the optimal POVM for MED of linearly independent states.
\begin{theorem}
\label{necsufcond}
In the problem of MED of a LI ensemble $\{p_i, \rho_i \}_{i=1}^{m}$ if a POVM, represented as $\{\Pi_i\}_{i=1}^{m}$, satisfies the following two conditions then it is the optimal POVM for MED of the said ensemble: 

\begin{enumerate}

\item  $\Pi_i \left( p_i \rho_i - p_j \rho_j \right) \Pi_j = 0, \; \forall \, 1 \leq i, \; j \leq m.$ This is equivalently expressed as: $ \left( Z - p_i \rho_i \right) \Pi_i = 0, \; \forall \; 1 \leq i \leq m$,
\label{ZGREATER0}
\item  $Z > 0$,
\end{enumerate}
where $Z$ is defined as in \eqref{Z}.
\end{theorem}
\begin{proof}
We need to prove that once we find $\{ \Pi_i \}_{i=1}^m$ which satisfies condition 1. and 2., i.e.,  such that conditions \eqref{cslack} (or equivalently equation \eqref{St}) and \eqref{Glb}, then that implies that $ \sum_{i=1}^{m} p_i \rho_i \Pi_i \geq p_i \rho_i, \; \forall \; 1 \leq i \leq m$. Suppose that 1. has been satisfied. This implies that we found a block diagonal matrix $D\geq0$ (given by equation \eqref{DX}) such that the block-diagonal of \emph{a} hermitian square root of $DGD$ (equation \eqref{DXG}) is $D^2$. The $i$-th block in this block-diagonal matrix $D$ is a positive semi-definite $r_i \times r_i$ matrix denoted by $X^{(ii)}$. Additionally, theorem \eqref{Xii} tells us that the non-singularity of $Z$ implies that $D$ has to be non-singular, i.e., $Det(Z) \neq 0 \Rightarrow Det(D) \neq 0 $. This is equivalent to the statement that $X^{(ii)}$ is of rank $r_i$, i.e., $X^{(ii)}>0, \; \forall \; 1 \leq i \leq m $. Using $X^{(ii)}$ define a new set of vectors as given in 
equations \eqref{chi}. Let's expand $Z$ and $p_i\rho_i$ in the operator basis 
$ \{ \ketbrat{\chi}{i_1j_1}{\chi}{i_2j_2} \; | \; 1 \leq i_1, i_2 \leq m, \; 1 \leq j_1 \leq r_{i_1} \text{ and }  1 \leq j_2 \leq r_{i_2}  \}$:
\begin{align}
\label{Zchi}
Z & = \sum_{i_1, \; i_2 =1}^{m} \sum_{j_1=1}^{r_{i_1}}\sum_{j_2=1}^{r_{i_2}} \left( W^\dag G^{-\frac{1}{2}} D^{-1} \right)^{(i_1 \; i_2)}_{j_1 \; j_2}\ketbrat{\chi}{i_1j_1}{\chi}{i_2j_2}\\
p_i\rho_i & = \sum_{\substack{k,l=1}}^{r_i} ({X^{(ii)}}^{-2})_{kl} \ketbrat{\chi}{ik}{\chi}{il} 
\end{align}

Thus $Z> 0 \Leftrightarrow W^\dag G^{-\frac{1}{2}}D^{-1} > 0 \Leftrightarrow DG^\frac{1}{2}W >0$.

Thus proving : $Z>0 \Rightarrow  Z \geq p_i \rho_i, \; \forall \; 1 \leq i \leq m$  is equivalent to proving $ W^\dag G^{-\frac{1}{2}} D^{-1} >0$ $ \Rightarrow W^\dag G^{-\frac{1}{2}} D^{-1} \geq $ $  \left( X^{(ii)} \right)^{-2}$, $\forall \; 1 \leq i \leq m$. Since $W^\dag G^{-\frac{1}{2}} D^{-1}$ $=$ $\left( D G^\frac{1}{2} W \right)^{-1}$, our objective is to prove that given $ \left( D G^\frac{1}{2} W \right)^{-1}  >0$ (where $D G^\frac{1}{2} W$ is given by equation \eqref{DXG}) implies that \footnotesize:
\begin{eqnarray} 
\label{invdepict1}
& \begin{pmatrix}
(X^{(11)})^2  & \cdots &  {X^{(11)}} X^{(1i)} & \cdots &  {X^{(11)}} X^{(1m)} \\
\vdots  & \ddots & \vdots & \ddots & \vdots \\
{X^{(ii)}} X^{(i1)}  & \cdots &  (X^{(ii)})^2 & \cdots &   {X^{(ii)}} X^{(im)}\\  
\vdots &  \ddots & \vdots & \ddots & \vdots \\
{X^{(mm)}} X^{(m1)}  & \cdots  &   {X^{(mm)}} X^{(mi)} & \cdots & (X^{(mm)})^2
\end{pmatrix}^{-1} \geq \begin{pmatrix}
0 & \cdots & 0 & \cdots & 0\\
\vdots  & \ddots & \vdots & \ddots & \vdots \\
0 & \cdots &  (X^{(ii)})^{-2} & \cdots & 0\\ 
\vdots  & \ddots & \vdots & \ddots & \vdots \\
0 & \cdots & 0 & \cdots & 0\\ 
\end{pmatrix}, \; \forall \; 1 \leq i \leq m
&  \notag \\
& \Bigg( \text{Permute: } \left\{ \begin{array} {l l l}
k \rightarrow & m+k-(i-1), &  \forall \; 1 \leq k \leq i-1 \notag \\
k \rightarrow & k-(i-1),   & \forall \; i \leq k \leq m \notag 
\end{array} 
\right. \Bigg)
\notag \\ 
&  \notag \\
\Longleftrightarrow & \begin{pmatrix}
(X^{(ii)})^2  &  X^{(ii)} X^{\scriptscriptstyle(i \: i+1)}  & \cdots &  X^{(ii)} X^{\scriptscriptstyle(i \: i-1)}\\
X^{\scriptscriptstyle( i\!{\scriptscriptstyle +}\!1 \: i\!{\scriptscriptstyle +}\!1)} X^{\scriptscriptstyle(i+1 \: i)}  &  {X^{\scriptscriptstyle( i+1 \: i+1)}}^{2}  &\cdots &   X^{\scriptscriptstyle(i+1 \: i+1)} X^{\scriptscriptstyle(i+1 \: i-1)}\\
\vdots& \vdots & \ddots & \vdots \\
X^{\scriptscriptstyle(i-1 \: i-1)} X^{\scriptscriptstyle(i-1 \: i)}& X^{\scriptscriptstyle( i-1 \: i-1)} X^{\scriptscriptstyle(i-1 \: i+1)}  & \cdots & (X^{\scriptscriptstyle(i-1 \: i-1)})^2
\end{pmatrix}^{-1}  \! {\scriptstyle \geq } \begin{pmatrix}
(X^{\scriptscriptstyle(ii)})^{\scriptscriptstyle{-2}}& 0 & \cdots &   0\\
0 &  0 & \cdots &  0\\ 
\vdots & \vdots & \ddots & \vdots \\
0 &  0 & \cdots &  0  
\end{pmatrix}, \; \forall \; i.
\end{eqnarray}
\normalsize
Define:\begin{align}
\left(
\begin{array}{c|c}
A       & \qquad B  \qquad  \qquad\\ \hline
~        & \qquad ~  \qquad  \qquad\\
B^{\dag} & \qquad C  \qquad  \qquad\\
~        & \qquad ~  \qquad  \qquad \end{array}
\right)
\equiv
\left(
\begin{array}{c|ccc}
(X^{(ii)})^2  &  X^{(ii)} X^{(i \: i+1)}  & \cdots &  X^{(ii)} X^{(i \: i-1)}\\ \hline
X^{( i+1 \: i+1)} X^{(i+1 \: i)}   &  {X^{( i+1 \: i+1)}}^{2}  &\cdots &   X^{(i+1 \: i+1)} X^{(i+1 \: i-1)}\\
\vdots  & \vdots & \ddots & \vdots \\
X^{(i-1 \: i-1)} X^{(i-1 \: i)}  & X^{( i-1 \: i-1)} X^{(i-1 \: i+1)}  & \cdots & (X^{(i-1 \: i-1)})^2
\end{array}
\right)
\end{align}
Hence our objective to prove that:
\begin{eqnarray}
\label{invdepict2}
\begin{pmatrix}
 A & B \\
B^{\dag} & C 
\end{pmatrix}^{-1}  > 0 \Longrightarrow
\begin{pmatrix}
 A & B \\
B^{\dag} & C 
\end{pmatrix}^{-1} \geq 
\begin{pmatrix}
 A^{-1} & 0 \\
0 & 0 
\end{pmatrix}
\end{eqnarray}
Given that $\bigl(\begin{smallmatrix}
A&B\\ B^{\dag}&C
\end{smallmatrix} \bigr) > 0$ its inverse is given by \cite{Boyd}: 
\begin{align}
 \begin{pmatrix}
  A & B \\
B^{\dag} & C 
 \end{pmatrix}
^{-1} & = 
\begin{pmatrix}
A^{-1} +  Q S_AQ^{\dag} & -Q S_A \\
-S_A Q^{\dag} & S_A 
\end{pmatrix}\\
& =  \begin{pmatrix}
   A^{-1} & 0 \\
0 & 0
  \end{pmatrix} + 
\begin{pmatrix}
 Q S_AQ^{\dag} & -Q S_A \\
-S_A Q^{\dag} & S_A 
\end{pmatrix}
\end{align}
where $S_A \equiv (C-B^{\dag}A^{-1}B)^{-1} > 0$ is the inverse of the Schur complement of $A$ in  $\bigl(\begin{smallmatrix} A&B\\ B^{\dag}&C \end{smallmatrix} \bigr)$ and $Q \equiv A^{-1}B$ \cite{Boyd}. Hence the inequality \eqref{invdepict2} amounts to proving the following:
\begin{eqnarray}
\label{invdepict3}
\begin{pmatrix}
 Q S_AQ^{\dag} & -Q S_A \\
-S_A Q^{\dag} & S_A 
\end{pmatrix} \geq 0
\end{eqnarray}
As shown in \cite{Boyd}, if $S_A >0 $, then: $\bigl(\begin{smallmatrix} Q S_AQ^{\dag} & -Q S_A \\ -S_A Q^{\dag} & S_A \end{smallmatrix} \bigr) \geq 0 \Leftrightarrow$ Schur complement of $S_A$ in $\bigl(\begin{smallmatrix} Q S_AQ^{\dag} & -Q S_A \\ -S_A Q^{\dag} & S_A \end{smallmatrix} \bigr) \geq 0$. Now $\bigl(\begin{smallmatrix}
A&B\\ B^{\dag}&C
\end{smallmatrix} \bigr) > 0 \Longrightarrow S_A > 0$. The Schur complement of $S_A$ in $\bigl(\begin{smallmatrix} Q S_AQ^{\dag} & -Q S_A \\ -S_A Q^{\dag} & S_A \end{smallmatrix} \bigr)$ is equal to $0$.  This implies that $\bigl(\begin{smallmatrix} Q S_AQ^{\dag} & -Q S_A \\ -S_A Q^{\dag} & S_A \end{smallmatrix} \bigr) \geq 0$. Hence the inequality \eqref{invdepict3} is true. This proves condition 1. of the theorem (or equivalently condition \eqref{St}) and $Z>0$ subsumes condition given by \eqref{Glb}. This proves the theorem. 
\end{proof}
Hence the necessary and sufficient conditions \eqref{St} (or equivalently equation \eqref{cslack}) and \eqref{Glb} are subsumed in the statement: $DG^{\frac{1}{2}}W >0$. Alternatively, the necessary and sufficient conditions can be put in the following corollary:

\begin{corollary}
\label{corollary1}
The necessary and sufficient condition for an $m$-element POVM $\{ \Pi_i \}_{i=1}^{m}$ to optimally discriminate among an ensemble of $m$ linearly independent states $\{ p_i, \; \rho_i\}_{i=1}^{m}$ is that $\{ \Pi_i \}_{i=1}^{m}$ is a projective measurment and $\sum_{i=1}^{m} p_i \rho_i \Pi_i > 0$.
\end{corollary}

We can re-express the necessary and sufficient conditions to obtain the optimal POVM for MED of the ensemble $\widetilde{P}$ as: 

\begin{itemize}
 \item[\textbf{A}:] \label{AA} One needs to find a block diagonal matrix, $D=$ $ Diag ($ $X^{(11)},$ $X^{(22)}$ $,\cdots,$ $X^{(mm)} ) $ $\geq0$ where $X^{(ii)}$ is an $r_i \times r_i$ positive definite matrix, so that the diagonal blocks of the positive square root of the matrix $DGD$ are given by $\left( X^{(11)} \right)^2, \left( X^{(22)} \right)^2, \cdots,\left( X^{(mm)} \right)^2$ respectively. Here $G$ corresponds to the gram matrix of vectors $\{ \tket{\psi}{ij_i} \; | \; 1 \leq i \leq m, \; 1 \leq j_i \leq r_i\}$ where $p_i\rho_i = \sum_{j_i=1}^{r_i} \ketbrat{\psi}{ij_i}{\psi}{ij_i}$, for all $i=1,2,\cdots,m$.

\end{itemize}

Condition \textbf{A} is a rotationally invariant form of expressing conditions \eqref{cslack} (or equivalently \eqref{St}) and \eqref{Glb}. 

We will now construct the ensemble $\widetilde{Q} = \{ q_i , \sigma_i \}_{i=1}^{m} \; \in \ens$, such that $supp \left( q_i \sigma_i \right) = supp \left(  p_i \rho_i \right), \; \forall \; 1 \leq i \leq m$, and for which the relation $PGM \left(\widetilde{Q}  \right) \{ \Pi_i \}_{i=1}^{m}$ holds true. 

Using equation \eqref{chi}, define the following: 

\begin{equation}
\label{sigma}
\sigma_i \equiv \frac{1}{\sum_{k_i=1}^{r_i} \tbraket{\chi}{i k_i}{\chi}{i k_i}}  \sum_{j_i=1}^{r_i} \ketbrat{\chi}{i j_i}{\chi}{i j_i}, \; \forall \; 1 \leq i \leq m,
\end{equation}

\begin{equation}
\label{qi}
q_i \equiv \dfrac{\sum_{j_i=1}^{r_i} \tbraket{\chi}{i j_i}{\chi}{i j_i}}{\sum_{l=1}^{m} \sum_{k_l=1}^{r_l} \tbraket{\chi}{l k_l}{\chi}{l k_l}} , \; \forall \; 1 \leq i \leq m.
\end{equation}

By the very definition $q_i > 0, \; \forall \; 1 \leq i \leq m$. And since the set $\{ \tket{\chi}{ij_i} \}_{j_i=1}^{r_i}$ spans $supp \left( p_i \rho_i \right)$, we have that $supp \left( q_i \sigma_i \right) = supp \left( p_i\rho_i \right), \; \forall \; 1 \leq i \leq m$. It remains to be shown that $\{ \Pi_i \}_{i=1}^{m}$ is the PGM of $\widetilde{Q}$.

\begin{theorem}
\label{PGMtheorem}
$\{ \Pi_i \}_{i=1}^{m}$ is the PGM of $\widetilde{Q}$, i.e., $ \Pi_i = \left( \sum_{j=1}^{m} q_j \sigma_j \right)^{- \frac{1}{2} } q_i \sigma_i  \left( \sum_{k=1}^{m} q_k \sigma_k \right)^{- \frac{1}{2} }, \; \forall \; 1 \leq i \leq m$.
\end{theorem}

\begin{proof}

We introduce a set of vectors complementary to the set $\{ \tket{\chi}{i j_i} \}_{i=1, \; j_i=1}^{i = m, \; j_i = r_i}$ in the same way that the vectors $\{ \tket{u}{i j_i} \}_{i=1, \; j_i=1}^{i = m, \; j_i = r_i}$ is complementary to the set $\{ \tket{\psi}{i j_i} \}_{i=1, \; j_i=1}^{i = m, \; j_i = r_i}$ based on equation \eqref{psiandu}.

\begin{equation}
 \label{chiandomega}
 \tbraket{\chi}{i_1j_1}{\omega}{i_2j_2} = \delta_{i_1 i_2}\delta_{j_1j_2}, \; \forall \; 1 \leq i_1, i_2 \leq m, \; 1 \leq j_1 \leq r_{i_1}, \; 1 \leq j_2 \leq r_{i_2}.
\end{equation}

Based on the definition of the vectors $\{ \tket{\chi}{i j_i} \}_{i=1, \; j_i=1}^{i = m, \; j_i = r_i}$ from equation \eqref{chi}:
\begin{equation}
\label{wu}
\tket{\omega}{ij_i} \equiv \sum_{l =1}^{m} \sum_{k_l=1}^{r_l} \left({D}^{-1}\right)^{(l \; i)}_{k_l \; j_i}\tket{u}{l k_l}
\end{equation}
From the definition of $\tket{\omega}{ij_i}$ in equations \eqref{chiandomega}, \eqref{wu} it is easy to see that $\{ \tket{\omega}{ij_i} \}_{i=1, \; j_i=1}^{i = m, \; j_i = r_i}$ will form a linearly independent set. We can expand $\ket{w_{ij}}$ from equation \eqref{Pidecomposition} in $\tket{\omega}{ij}$:
\begin{equation}
\label{mixedv3}
\ket{w_{ij}} = \sum_{l=1}^{m}\sum_{k_l=1}^{r_l} \left( D G^{\frac{1}{2}}W \right)^{(l \; i)}_{k_l \; j}\tket{\omega}{l k_l}, 
\end{equation} and, similar to equation \eqref{mixedP} we get: 

\begin{equation}
\label{Pomega}
\Pi_i = \sum_{l_1, l_2 = 1}^{m} \sum_{k_1=1}^{r_1} \sum_{k_2=1}^{r_2} \left( \sum_{j=1}^{r_i} \left( D G^{\frac{1}{2}}W \right)^{(l_1 \; i)}_{k_1  \;  j}  \left( W^\dag G^{\frac{1}{2}} D \right)^{(i  \;  l_2)}_{j  \; k_2} \right) \ketbrat{\omega}{l_1 k_1}{\omega}{l_2 k_2}.
\end{equation}

We will prove that $\left( \sum_{j=1}^{m} q_j \sigma_j \right)^{- \frac{1}{2} } q_i \sigma_i  \left( \sum_{k=1}^{m} q_k \sigma_k \right)^{- \frac{1}{2} }$ is equal to the RHS of equation \eqref{Pomega}, $\forall \; 1 \leq i \leq m$. That will prove the theorem.

By the definition of $\sigma_i$ in equation \eqref{sigma} we get that $ \sum_{i=1}^{m} q_i \sigma_i $ is given by:

\begin{equation}
\label{sumsigma}
\sum_{i=1}^{m} q_i \sigma_i = \frac{1}{\sum_{s=1}^{m} \sum_{t_s=1}^{r_s} \tbraket{\chi}{s t_s}{\chi}{s t_s}} \sum_{l=1}^{m} \sum_{k_l=1}^{r_l} \ketbrat{\chi}{l k_l}{\chi}{l k_l}
\end{equation}

Using equation \eqref{sumsigma}, it can easily be verified that:

\begin{equation}
\label{sumsigmainv}
\left( \sum_{i=1}^{m} q_i \sigma_i \right)^{-1} = \left( \sum_{s=1}^{m} \sum_{t_s=1}^{r_s} \tbraket{\chi}{s t_s}{\chi}{s t_s} \right)  \sum_{l=1}^{m} \sum_{k_l=1}^{r_l} \ketbrat{\omega}{l k_l}{\omega}{l k_l}
\end{equation}

Bearing in mind the $DG^\frac{1}{2} W$ is the positive square root of the matrix $DGD$, and that $DGD$ is the gram matrix of the set of vectors $\{ \tket{\chi}{i j_i} \}_{i=1, \; j_i=1}^{i = m, \; j_i = r_i}$, it can be easily verified that:

\begin{equation}
\label{sumsigmainvsq}
\left( \sum_{i=1}^{m} q_i \sigma_i \right)^{-\frac{1}{2}} = \left( \sum_{s=1}^{m} \sum_{t_s=1}^{r_s} \tbraket{\chi}{s t_s}{\chi}{s t_s} \right)^\frac{1}{2}  \sum_{l_1=1}^{m} \sum_{k_1=1}^{r_{l_1}} \sum_{l_2=1}^{m} \sum_{k_2=1}^{r_{l_2}} \ketbrat{\omega}{l_1 k_1}{\omega}{l_2 k_2} \left( D G^\frac{1}{2} W \right)^{ \left( l_1 \; l_2 \right) }_{k_1 \; k_2}.
\end{equation}

Using the expression for $\left( \sum_{i=1}^{m} q_i \sigma_i \right)^{-\frac{1}{2}}$ in equation \eqref{sumsigmainvsq}, the expression for $q_i \sigma_i$ in equations \eqref{sigma} and \eqref{qi} and after a bit of algebra we get the result that $\left( \sum_{j=1}^{m} q_j \sigma_j \right)^{- \frac{1}{2} } q_i \sigma_i  \left( \sum_{k=1}^{m} q_k \sigma_k \right)^{- \frac{1}{2} }$ is equal to the RHS of equation \eqref{Pomega}, $\forall \; 1 \leq i \leq m$. This establishes that $\{ \Pi_i \}_{i=1}^{m} = PGM (\widetilde{Q})$. Hence proved.
\end{proof}

Thus we have shown that for every $\widetilde{P} = \{ p_i, \rho_i \}_{i=1}^{m} \in \ens$ there exists an ensemble $\widetilde{Q} = \{ q_i, \sigma_i \}_{i=1}^{m} \in \ens$ such that $supp \left( q_i \sigma_i \right) = supp \left( p_i \rho_i \right), \; \forall \; 1 \leq i \leq m$ and such that $\widetilde{Q}$'s PGM is $\{ \Pi_i \}_{i=1}^{m} = \mathscr{P}\left( \widetilde{Q} \right)$. This establishes the $\widetilde{P} \longrightarrow \widetilde{Q}$ correspondence mentioned in the end of the previous subsection. 

The next question that needs to be answered is whether there was any ambiguity in the way we arrived at the ensemble $\widetilde{Q}$ for a given $\widetilde{P}$? The only ambiguity that we have allowed to remain is in the choice of the decomposition of the states $p_i \rho_i$ in the pure unnormalized states $\tket{\psi}{i j_i}$ in equation \eqref{rhodecomposition}. For a given choice of such a decomposition for all $i = 1, 2, \cdots, m$, we arrived at a unique $n \times n$ unitary $W$ such that the block diagonal matrix $D$, defined in equation \eqref{part1} and equation \eqref{DX}, is positive definite. And using the $X^{(ii)}$ matrices we arrived at the set of states $\tket{\chi}{i j_i}$ in equation \eqref{chi} from which the states $q_i \sigma_i$ were constructed using equations \eqref{sigma} and \eqref{qi}. It is now natural to ask if the final states $q_i \sigma_i$ depend on the choice of the decomposition of the $p_i \rho_i$'s used in equation \eqref{rhodecomposition}. Very briefly we take the reader 
through the sequence of steps that show that this isn't the case.

Let $U'^{(i)}$ be an $r_i \times r_i$ unitary, for $i=1, 2, \cdots, m$. Arrange the $m$ unitary matrices - $U'^{(1)}$, $U'^{(2)}$, $\cdots$, $U'^{(m)}$ as diagonal blocks of an $n \times n$ unitary matrix which we call $U'_D$:

\begin{equation}
\label{U'D}
U'_D = \begin{pmatrix}
{U'}^{(1)} & 0 & \cdots & 0 \\
0 & {U'}^{(2)} & \cdots & 0 \\
\vdots & \vdots & \ddots & \vdots\\
0 & 0 & \cdots & {U'}^{(m)}
\end{pmatrix}.
\end{equation} 

Define the following:

\begin{eqnarray}
\label{psi'}
\tket{\psi '}{i j_i} \equiv \sum_{l=1}^{m} \sum_{k_l =1 }^{r_l} \left( U'_D \right)^{\left( l \; i \right)}_{k_l \; j_i} \tket{\psi}{i j_i}, \; \forall \; 1 \leq i \leq m, 1 \leq j_i \leq r_i, \\ 
\label{u'}
\tket{u '}{i j_i} \equiv \sum_{l=1}^{m} \sum_{k_l =1 }^{r_l} \left( U'_D \right)^{\left( l \; i \right)}_{k_l \; j_i} \tket{u}{i j_i},  \; \forall \; 1 \leq i \leq m, 1 \leq j_i \leq r_i.
\end{eqnarray}

Note that $p_i \rho_i = \sum_{j=1}^{r_i} \ketbrat{\psi '}{i j_i}{\psi '}{i j_i}$, $\forall \; 1 \leq i \leq m$, which implies that we now have an alternative decomposition of the states $p_i \rho_i$ into the pure states $\tket{\psi}{i j_i}$.  Also note that:

\begin{equation}
\label{psi'andu'}
\tbraket{\psi'}{i_1j_1}{u'}{i_2j_2} = \delta_{i_1 i_2} \delta_{j_1 j_2}, \; \forall \; 1 \leq i_1,i_2 \leq m \text{ and } 1 \leq j_1 \leq r_{i_1}, \; 1 \leq j_2 \leq r_{i_2},
\end{equation}

which is similar to equation \eqref{psiandu}. 

Equation \eqref{wexpandu} modifies to:

\begin{equation} 
\label{wexpandu'}
\ket{w_{ij_i}} = \sum_{l=1}^{m}\sum_{k_l=1}^{r_l} \left( {U'_D}^\dag G^\frac{1}{2} W U'_D \right)^{(l \; i)}_{k_l \; j_i} \tket{u'}{lk_l}, \; \forall \; 1 \leq i \leq m, \; 1 \leq j_i \leq r_i.
\end{equation} 

Earlier on, we chose the $n \times n$ unitary $W$ in such a manner that the diagonal blocks of $G^\frac{1}{2}W$, i.e., the matrices $X^{(11)}$, $X^{(22)}$, $\cdots$, $X^{(mm)}$ are hermitian (and positive definite). The diagonal blocks now become ${U'^{(1)}}^\dag X^{(11)}$, ${U'^{(2)}}^\dag X^{(22)}$, $\cdots$, ${U'^{(m)}}^\dag X^{(mm)}$. Hence we now employ a different decomposition for the projectors $\Pi_i$ than given in equation \eqref{wexpandu}:

\begin{equation} 
\label{w'expandu'}
\ket{w'_{ij_i}} = \sum_{l=1}^{m}\sum_{k_l=1}^{r_l} \left( {U'_D}^\dag G^\frac{1}{2} W U'_D \right)^{(l \; i)}_{k_l \; j_i} \tket{u'}{lk_l}, \; \forall \; 1 \leq i \leq m, \; 1 \leq j_i \leq r_i.
\end{equation} 

The diagonal blocks in this case are ${U'^{(1)}}^\dag X^{(11)} U'^{(1)}$, ${U'^{(2)}}\dag X^{(22)} U'^{(2)}$, $\cdots$, ${U'^{(m)}}\dag X^{(mm)} U'^{(m)}$, which are not only hermitian but positive definite (since $X^{(ii)} > 0, \; \forall \; 1 \leq i \leq m $).

Just in the case of equation \eqref{chi}, define:

\begin{eqnarray}
\label{chi'}
& \tket{\chi'}{ij_i} & =  \sum_{\substack{k=1}}^{r_i}  \left( {U'^{(i)}}^\dag X^{(ii)} U'^{(i)}  \right)_{kj} \tket{\psi '}{ik} \notag \\
& ~ & =   \sum_{\substack{k=1}}^{r_i}  \left( X^{(ii)} U'^{(i)}  \right)_{kj} \tket{\psi }{ik}, \forall \; 1 \leq i \leq m, \; 1 \leq j_i \leq r_i.
\end{eqnarray}

Using equation \eqref{chi'} and equation \eqref{chi} it isn't difficult to show that:

\begin{equation}
\label{chichi'}
\sum_{j=1}^{r_i} \ketbrat{\chi '}{i j}{\chi '}{i j} = \sum_{k=1}^{r_i} \ketbrat{\chi }{i k}{\chi }{i k}, \; \forall \; 1 \leq i \leq m.
\end{equation}

Using equation \eqref{sigma} and \eqref{qi} we get that:

\begin{equation}
\label{sigma'}
\sigma_i = \frac{1}{\sum_{k_i=1}^{r_i} \tbraket{\chi '}{i k_i}{\chi '}{i k_i}}  \sum_{j_i=1}^{r_i} \ketbrat{\chi '}{i j_i}{\chi '}{i j_i}, \; \forall \; 1 \leq i \leq m.
\end{equation}

\begin{equation}
\label{q'i}
q'_i \equiv \dfrac{\sum_{j_i=1}^{r_i} \tbraket{\chi '}{i j_i}{\chi '}{i j_i}}{\sum_{l=1}^{m} \sum_{k_l=1}^{r_l} \tbraket{\chi '}{l k_l}{\chi '}{l k_l}} , \; \forall \; 1 \leq i \leq m.
\end{equation}

This establishes that the correspondence $\widetilde{P} \longrightarrow \widetilde{Q}$ is invariant over the choice of pure state decompositions of $p_i\rho_i$\footnote{Actually, this association is also invariant over the choice of spectral decomposition of $\Pi_i$ in equation \eqref{Pidecomposition}. Our choice of spectral decomposition was such that the $D$ matrix, defined in equation \eqref{DX}, is positive definite for the sake of the convenience this offers; this isn't necessary.}. Going through all the steps taken to construct th ensemble $\widetilde{Q}$ from the ensemble $\widetilde{P}$ and the optimal POVM $\mathscr{P}\left( \widetilde{Q} \right) = \{ \Pi_i \}_{i=1}^{m}$, we can see that there is no degree of freedom on account of which the association of $\widetilde{P}$ to $\widetilde{Q}$ can be regarded as ambiguous. This tells us that the correspondence $\widetilde{P} \longrightarrow \widetilde{Q}$ is a map from $\ens$ to itself. We denote this map by $\mathscr{R}$; thus we have $\mathscr{R} : \
ens \longrightarrow \ens$, such that $\mathscr{R} \left( \widetilde{P} \right) = \widetilde{Q}$ and such that $\mathscr{P}\left(\widetilde{P} \right) = PGM\left(\mathscr{R}\left(\widetilde{P}\right)\right)$. 

\subsection{Invertibility of $\mathscr{R}$}
\label{invertibilityofR}

The existence of the map $\mathscr{R}$ was already demonstrated in \cite{Mas}. The reason we went through the elaborate process of re-demonstrating its existence is that these sequence of steps enables us to trivially establish that the map $\mathscr{R}$ is invertible.

We first show that $\mathscr{R}$ is onto. 
\begin{theorem}
\label{Ronto}
The map $\mathscr{R}$ is onto.
\end{theorem}

\begin{proof} 

This means we have to prove that $\forall \; \widetilde{Q} \in \ens,$  $ \exists \; \text{some } \widetilde{P} \in \ens$ $\ni \; \mathscr{R} \left( \widetilde{P} \right) = \widetilde{Q}$. 

Let $\widetilde{Q} = \{ q_i, \sigma_i \}_{i=1}^{m} \in \ens$. Thus $supp \left( q_1 \rho_1 \right)$, $supp \left( q_2 \rho_2 \right)$, $\cdots$, $supp \left( q_m \rho_m \right)$ are LI subspaces of $\mathcal{H}$ of dimensions $r_1, r_2, \cdots, r_m$ respectively. Let the following be a resolution of the state $q_i \sigma_i$ into pure states:

\begin{equation}
\label{sigmaresolution}
q_i \sigma_i = \sum_{j=1}^{r_i} \ketbrat{\zeta}{ij_i}{\zeta}{ij_i}, \; \forall \; 1 \leq i \leq m.
\end{equation}

There is a $U \left( r_i \right)$ degree of freedom of choosing such a resolution. The set $\{ \tket{\zeta}{ij_i} \}_{i=1, \; j_i = 1}^{i=m, \; j_i = r_i}$ is LI. Let's denote the gram matrix corresponding to the set of states $\{ \tket{\zeta}{ij_i} \}_{i=1, \; j_i = 1}^{i=m, \; j_i = r_i}$ by $F$. The matrix elements of $F$ are given by:

\begin{equation}
\label{Fmatrixelement}
F^{(i_1 \; i_2)}_{j_1 \; j_2} = \tbraket{\zeta}{i_1j_1}{\zeta}{i_2j_2}, \; \forall \; 1 \leq i_1, i_2 \leq m, \; 1 \leq j_1 \leq r_{i_1}, \; 1 \leq j_2 \leq r_{i_2}.
\end{equation}

$F^{\frac{1}{2}}$ is the positive definite square root of $F$. Partition $F^{\frac{1}{2}}$ in the following manner: 

\begin{equation}
\label{Froot}
F^{\frac{1}{2}}= \begin{pmatrix}
                  H^{(11)} &  H^{(12)} & \cdots & H^{(1m)} \\
                  H^{(21)} &  H^{(22)} & \cdots & H^{(2m)} \\
                  \vdots & \vdots & \ddots & \vdots \\
                  H^{(m1)} &  H^{(m2)} & \cdots & H^{(mm)} 
                 \end{pmatrix},
\end{equation} where $H^{(ij)}$ is the $\left(i, j \right)$-th block matrix in $F$ and is of dimension $r_i \times r_j$, $\forall \; 1 \leq i, j \leq m$. Note that $F^{\frac{1}{2}} > 0 $ implies that $H^{(ii)} > 0, \; \forall \; 1 \leq i \leq m$.

Corresponding to the set $\{ \tket{\zeta}{ij_i} \}_{i=1, \; j_i = 1}^{i=m, \; j_i = r_i}$ $\exists$ another unique set $\{ \tket{z}{ij_i} \}_{i=1, \; j_i = 1}^{i=m, \; j_i = r_i}$  such that

\begin{equation}
\label{zetaz}
\tbraket{\zeta}{i_1j_1}{z}{i_2j_2} = \delta_{i_1, i_2}\delta_{j_1,j_2}, \; \forall \; 1 \leq i_1, \; \leq i_2 \leq m, \; 1 \leq j_1 \leq r_{i_1}, \; 1 \leq j_2 \leq r_{i_2}.
\end{equation}

The relation that the set $\{ \tket{z}{ij_i} \}_{i=1, j_i = 1}^{i=m, j_i = r_i}$ bears to $\{ \tket{\zeta}{ij_i} \}_{i=1, j_i = 1}^{i=m, j_i = r_i}$ is equivalent to that which $\{ \tket{u}{ij_i} \}_{i=1, j_i = 1}^{i=m, j_i = r_i}$ bears to $\{ \tket{\psi}{ij_i} \}_{i=1, j_i = 1}^{i=m, j_i = r_i}$ (see equation \eqref{psiandu}); or as  $\{ \tket{\omega}{ij_i} \}_{i=1, j_i = 1}^{i=m, j_i = r_i}$ bears to $\{ \tket{\chi}{ij_i} \}_{i=1, j_i = 1}^{i=m, j_i = r_i}$ (see equation \eqref{chiandomega}). 

Let the PGM of $\{q_i, \; \sigma_i \}_{i=1}^{m}$ be denoted by $\{ \Omega_i \}_{i=1}^{m}$. Thus $\Omega_i \geq 0$ and $\sum_{i=1}^{m} \Omega_i \, = \, \mathbb{1}$. In the body of the proof of theorem \eqref{PGMtheorem} we constructed the PGM for an ensemble of mixed states using the pure state decomposition of the corresponding mixed states. Following the same sequence of steps gives us the $\Omega_i$ projectors  expanded in the $\{ \ketbrat{z}{l_1 k_1}{z}{l_2k_2} \; | \; 1 \leq l_1,l_2 \leq m, \; 1 \leq l_1 \leq r_{l_1}, \; 1 \leq k_2 \leq r_{l_2} \}$ operator basis:
 
\begin{equation}
\label{omegaexp}
\Omega_i = \sum_{l_1, l_2 =1}^{m}\sum_{k_1=1}^{r_{l_1}}\sum_{k_2=1}^{r_{l_2}}
\left( \sum_{j=1}^{r_i} \left( F^{\frac{1}{2}} \right)^{(l_1 \; i)}_{k_1 \; j} \left( F^{\frac{1}{2}} \right)^{(i \; l_2)}_{j \; k_2} \right) \ketbrat{z}{l_1k_1}{z}{l_2k_2}, \; \forall \; 1 \leq i \leq m.
\end{equation}

The gram matrix of the set $\{ \tket{z}{ij_i} \}_{i=1, j_i = 1}^{i=m, j_i = r_i}$ is given by $F^{-1}$ and using this fact it is trivial to show that the operators $\Omega_i$, given in equation \eqref{omegaexp}, are indeed projectors. Thus we have the PGM of the ensemble $\widetilde{Q}$ with us. Now we construct the ensemble which we will denote by $\widetilde{P'} = \{p'_i, \rho'_i \}_{i=1}^{m}$. This ensemble will be such that $\mathscr{R} \left( \widetilde{P}' \right) = \widetilde{Q}.$

Define the following:

\begin{eqnarray}
& \tket{\phi}{ij} & \equiv \sum_{k=1}^{r_i} \left( \left( H^{(ii)} \right)^{-\frac{1}{2}} \right)_{kj} \tket{\zeta}{ik}, \; \forall \; 1 \leq i \leq m, \\ 
\label{p'_irho'_i}
& p'_i \rho'_i & \equiv \frac{1}{\sum_{l=1}^{m} \sum_{k_l=1}^{r_l} \tbraket{\phi}{l k_l}{\phi}{l k_l}}  \sum_{j_i=1}^{r_i} \ketbrat{\phi}{i j_i}{\phi}{i j_i}, \; \forall \; 1 \leq i \leq m.
\end{eqnarray}

Note that $supp \left(p_i' \rho_i' \right) = supp \left(p_i \rho_i \right), \; \forall \; 1 \leq i \leq m$. This also implies that $\widetilde{P}' \in \ens$.

Let's denote $c = \frac{1}{\sum_{l=1}^{m} \sum_{k_l=1}^{r_l} \tbraket{\phi}{l k_l}{\phi}{l k_l}}$. We insert  equations \eqref{p'_irho'_i} and \eqref{omegaexp} into equation \eqref{Z} to obtain:

\begin{eqnarray}
& Z & = \sum_{\substack{i=1}}^{m} p'_i \rho'_i \Omega_i \notag \\
& ~ & = c \sum_{i_1,i_2=1}^{m}\sum_{j_1=1}^{r_{i_1}} \sum_{j_2=1}^{r_{i_2}} \left( F^{-\frac{1}{2}} \right)^{(i_1 \; i_2)}_{j_1 \; j_2}\ketbrat{\zeta}{i_1j_1}{\zeta}{i_2j_2} \; > \; 0
\end{eqnarray}

$PGM \left( \widetilde{Q} \right) = \{ \Omega_i \}_{i=1}^{m}$ is a projective measurment and $Z = \sum_{i=1}^{m} p_i \rho_i \Omega_i >0$. By the corollary \eqref{corollary1}, $PGM \left( \widetilde{Q} \right) = \mathscr{P} \left( \widetilde{P'} \right)$. We still need to verify if $\mathscr{R}\left(\widetilde{P}'\right) = \widetilde{Q}$ or not. To this purpose we need to construct the ensemble $\widetilde{Q}'$ from $\widetilde{P}'$ in the same way as $\widetilde{Q}$ was constructed from $\widetilde{P}$ in section \eqref{PQcorr}. Let's start by defining:

\begin{equation}
\label{DA}
D_A \equiv \begin{pmatrix}
            \left( H^{(11)} \right)^{-\frac{1}{2}} & 0 &  \cdots & 0 \\
            0 & \left( H^{(22)} \right)^{-\frac{1}{2}} &  \cdots & 0\\
            \vdots & \vdots & \ddots & \vdots \\
            0 & 0 & \cdots & \left( H^{(mm)} \right)^{-\frac{1}{2}}
           \end{pmatrix}
\end{equation}

From equation \eqref{p'_irho'_i} we see that the vectors $\{ \tket{\phi}{i j_i} \}_{j_i=1}^{r_i}$ form a resolution of the state $p'_i \rho'_i$. The set of vectors $\{ \tket{\phi}{i j_i} \}_{i=1,\;j_i=1}^{i=m, \; j_i=r_i}$ are LI, so the gram matrix associated with this set, which we denote by $G'$, must be positive definite. Indeed it is given by $G' = c D_A F D_A$ which is positive definite. The matrix equivalent of $G^\frac{1}{2} W$, given in equation \eqref{part1} , in this case is $\sqrt{c} D_A F^\frac{1}{2}$. Note that, upto unitary degree of freedom in the choice of the decomposition of the states $p'_i \rho'_i$ into pure unnormalized states $\tket{\phi}{i j_i}$, the matrix $\sqrt{c} D_A F^\frac{1}{2}$ can be uniquely associated with the ensemble $\{p'_i, \rho'_i \}_{i=1}^{m}, \; \forall \; 1 \leq i \leq m$. The diagonal blocks of $\sqrt{c} D_A F^\frac{1}{2}$ are $\sqrt{c} \left( H^{(11)} \right)^\frac{1}{2}, \; \sqrt{c} \left( H^{(22)} \right)^\frac{1}{2}, \; \cdots, \;\sqrt{c} \left( H^{(mm)} \right)
^\frac{1}{2}$. Hence, the role played by $D>0$, given in equation \eqref{DX}, here is $\sqrt{c} \left(D_A \right)^{-1}$. Thus the matrix equivalent of $D G^\frac{1}{2} W$, given in equation \eqref{DXG}, here is $c F^\frac{1}{2}$ which is positive definite, and whose block diagonals - $c H^{(11)}, \; c H^{(22)}, \cdots, c H^{(mm)}$, are squares of the block diagonals of the matrix $\sqrt{c} D_A F^\frac{1}{2}$. We can construct a new set of vectors $\{ \tket{\zeta'}{ij_i} \}_{i=1,j_i=1}^{i=m,j_i=r_i}$ from $\{ \tket{\phi}{ij_i} \}_{i=1,j_i=1}^{i=m,j_i=r_i}$ in the same way $\{ \tket{\chi}{ij_i} \}_{i=1,j_i=1}^{i=m,j_i=r_i}$ were constructed  from $\{ \tket{\psi}{ij_i} \}_{i=1,j_i=1}^{i=m,j_i=r_i}$ in equation \eqref{chi}; the role of $X^{(ii)}$ being played by $\sqrt{c}\left(H^{(ii)}\right)^\frac{1}{2}$. But then we get that $\tket{\zeta'}{ij_i} = \tket{\zeta}{ij_i}, \; \forall \; 1 \leq i \leq m, \; 1 \leq j_i \leq r_i$. This tells us that $q'_i \sigma'_i = \sum_{j=1}^{r_i} \ketbrat{\zeta'}{ij_i}{\zeta'}{ij_i}
$, $\forall \; 1 \leq i \leq m$. This shows us that $\mathscr{R}\left( \widetilde{P'} \right) = \widetilde{Q}$ is indeed true. Hence $\mathscr{R}$ is onto. \end{proof}

We next prove that $\mathscr{R}$ is one-to-one. 

\begin{theorem}
\label{Roneone}
$\mathscr{R}$ is one-to-one. 
\end{theorem}

\begin{proof}
We need to prove that if $ \mathscr{R} \left( \widetilde{P} \right) = \mathscr{R} \left( \widetilde{P'} \right) $ then $\widetilde{P} = \widetilde{P'}$, $\forall \; \widetilde{P}, \widetilde{P'} \in \ens$. Let's denote $\widetilde{Q} =\mathscr{R} \left( \widetilde{P} \right) =\{ q_i, \sigma_i \}_{i=1}^{m} $ and $\widetilde{Q'} =\mathscr{R} \left( \widetilde{P'} \right)=\{ q'_i, \sigma'_i \}_{i=1}^{m} $. Let $\widetilde{P} = \{ p_i, \rho_i \}_{i=1}^{m} $ and $\widetilde{P'} = \{ p'_i, \rho'_i \}_{i=1}^{m} $. 

Given that $\mathscr{R}\left( \widetilde{P} \right) = \widetilde{Q}$. This implies the following: for any pure state decomposition of the states $\{ q_i \sigma_i \}_{i=1}^{m}$, with a corresponding gram matrix $F$, there exists a corresponding pure state decomposition of the states $\{ p_i \rho_i \}_{i=1}^{m}$, with a corresponding gram matrix $G$, such that $G = c D_{A} F D_{A}$, where $D_A$ is as defined in equation \eqref{DA} and $F^\frac{1}{2}$ is as defined in equation \eqref{Froot} and $c$ being the normalization constant. 

Similarly, given that $\mathscr{R}\left( \widetilde{P'} \right) = \widetilde{Q'}$, any pure state decomposition of the states $\{ q'_i \sigma'_i \}_{i=1}^{m}$, with a corresponding gram matrix $F'$, there exists a corresponding pure state decomposition of the states $\{ p'_i \rho'_i \}_{i=1}^{m}$, with a corresponding gram matrix $G'$, such that $G' = c' {D'}_{A} F {D'}_{A}$, where all the primed quantities ${D'}_A$ and ${F'}^\frac{1}{2}$ are defined similar to unprimed quantities in the equations \eqref{DA} and \eqref{Froot} and $c'$ is the corresponding normalization constant.

That $\widetilde{Q}_1 = \widetilde{Q}_2$ implies that for any choice of pure state decomposition of the primed and unprimed ensemble states, there exists a block-diagonal unitary $U_D$ of the form given in equation \eqref{UD}, such that the gram matrices $F$ and $F'$ can be related by the relation: $F' = {U_D}^\dag F {U_D}$. It also implies that ${F'}^\frac{1}{2} = {U_D}^\dag F^\frac{1}{2} {U_D}$, ${D'}_A = {U_D}^\dag D_A {U_D}$. Thus we get the relation that $G' = {U_D}^\dag G {U_D}$. Thus the corresponding pure state decompositions of $\widetilde{P}$ and $\widetilde{P'}$ are related through an equation similar to equation \eqref{psi'} which implies that $\widetilde{P} = \widetilde{P'}$. 

Hence we have proved that $\mathscr{R}\left(\widetilde{P}'\right)=\mathscr{R}\left(\widetilde{P}\right)$ $\Longleftrightarrow \widetilde{P}'=\widetilde{P}$. Hence $\mathscr{R}$ is one to one. 
\end{proof}

The theorems \eqref{Roneone} and \eqref{Ronto} jointly establish that the map $\mathscr{R}$ is invertible. We summarize all that we have done in this section in the following:

\textbf{Hence we have proved the existence of a bijective function $\mathbf{\mathscr{R}: \ens \longrightarrow \ens}$ such that the optimal POVM for the MED of any LI ensemble $\mathbf{\widetilde{P} \in \ens}$, which is given by $\mathbf{\mathscr{P}\left(\widetilde{P}\right)}$, satisfies the following relation:}

\begin{equation}
\label{SOLFORM}
\mathbf{\mathscr{P}\left( \widetilde{P} \right) = PGM \left( \mathscr{R} \left( \widetilde{P} \right) \right).}
\end{equation}

\textbf{The inverse map $\mathbf{\mathscr{R}^{-1}}$ has an analytic expression:}

\begin{equation}
\label{invR}
\mathbf{\mathscr{R}^{-1}\left( \{ q_i, \sigma_i \}_{i=1}^{m} \right) = \{p_i, \rho_i \}_{i=1}^{m},}
\end{equation}

\textbf{where, if } \begin{itemize}
\item $\mathbf{q_i \sigma_i = \dfrac{1}{\sum_{s=1}^{m}\sum_{t_s=1}^{r_s} \tbraket{\chi}{st_s}{\chi}{st_s}} \sum_{l=1}^{m}\sum_{k_l=1}^{r_l} \ketbrat{\chi}{l k_l}{\chi}{l k_l}}$ 
\item $\mathbf{p_i \rho_i =  \dfrac{1}{\sum_{s=1}^{m}\sum_{t_s=1}^{r_s} \tbraket{\psi}{st_s}{\psi}{st_s}}\sum_{l=1}^{m}\sum_{k_l=1}^{r_l} \ketbrat{\psi}{l k_l}{\psi}{l k_l}}$
\end{itemize}

\textbf{are pure state decompositions of the states in} $\mathbf{\widetilde{Q}}$ \textbf{and} $\mathbf{\widetilde{P}}$\textbf{ respectively, then}  $\mathbf{\{ \tket{\chi}{i j_i} \}_{i=1, \; j_i =1}^{i=m, \; j_i = r_i}}$\textbf{ and }$\mathbf{\{ \tket{\psi}{i j_i} \}_{i=1, \; j_i =1}^{i=m, \; j_i = r_i}}$\textbf{ are related through the transformation:}

\begin{equation}
\label{psichi}
\mathbf{\tket{\psi}{i j} = c \sum_{k=1}^{r_i} \left( \left(H^{(ii)}\right)^{-\frac{1}{2}} \right)_{k j} \tket{\chi}{ik}, \; \forall \; 1 \leq i \leq m, \; 1 \leq j \leq r_i,}
\end{equation} \textbf{where $\mathbf{c = \dfrac{1}{\sqrt{\sum_{s=1}^{m}\sum_{t,t_1,t_2 = 1}^{r_s} \left( \left( H^{(ss)} \right)^{-\frac{1}{2}} \right)_{t t_1} \left( F \right)^{(s \; s)}_{t_1 \; t_2}  \left( \left( H^{(ss)} \right)^{-\frac{1}{2}} \right)_{t_2 t}}}}$ and where $\mathbf{F}$ is the gram matrix of the set $\mathbf{\{ \tket{\chi}{i j_i} \}_{i=1, \; j_i =1}^{i=m, \; j_i = r_i}}$ and $\mathbf{H^{(ii)}}$'s are as defined in equation \eqref{Froot}.} 

\section{Comparing MED for Mixed LI ensembles and LI pure state ensembles}
\label{compareMEDP}

Minimum Error Discrimination is the task of extracting information about a state, by discarding some of the uncertainty of which state Alice sends Bob from the ensemble. Heuristically, one can expect that Bob is required to extract \emph{more} information while performing MED of an ensemble of $n$ LI pure states, which span $\mathcal{H}$, compared to an ensemble of $m$ ($m<n$) LI mixed states, where the supports of these $m$ states also span $\mathcal{H}$. This is because Bob requires to ``probe" the first ensemble ``deeper" compared to the second ensemble of states. This is better appreciated when comparing the MED of a mixed state ensemble and an ensemble of LI pure states which form pure state decompositions of the mixed states in the former. In this case it is a natural to ask if, generally, the optimal POVM for the LI pure state ensemble is a pure state decomposition of the optimal POVM for the mixed state ensemble, i.e., when a mixed state ensemble $\{ p_i, \rho_i \}_{i=1}^{m}$, with optimal POVM $\{ \
Pi_i \}_{i=1}^{m}$, and a pure state ensemble $\{ \lambda_{ij_i}, \ketbra{\psi'_{ij_i}}{\psi'_{ij_i}} \}_{i=1,j_i=1}^{i=m, j_i=r_i}$, with optimal POVM $\{ \ketbra{w'_{ij_i}}{w'_{ij_i}}\}_{i=1,j_i=1}^{i=m, j_i=r_i}$, are related by $p_i \rho_i = \sum_{j=1}^{r_i} \lambda_{ij} \ketbra{\psi'_{ij}}{\psi'{ij}}$ is it generally true that $\Pi_i = \sum_{j=1}^{r_i} \ketbra{w'_{ij}}{w'_{ij}}$, $\forall \; 1 \leq i \leq m$? In general, the answer is no. But we will now show that for every LI mixed state ensemble, one can find a corresponding pure state decomposition such that the optimal POVM for the MED of the ensemble of these LI pure states is a pure state decomposition of the optimal POVM for MED of the mixed state ensemble. 

Let equation \eqref{rhodecomposition} give a pure state decomposition of $p_i \rho_i, \; \forall \; 1 \leq i \leq m$. Then corresponding to the states $\{ \tket{\psi}{ij_i} \}_{i=1, j_i = 1}^{i=m, j_i = r_i}$, there exist a unique set of states $\{ \tket{u}{ij_i} \}_{i=1, j_i = 1}^{i=m, j_i = r_i}$, given by equation \eqref{u}, and a unique $n \times n$ unitary $W$, such that the projectors of the optimal POVM for the ensemble $\{p_i, \rho_i \}_{i=1}^{m}$ are given by equation \eqref{mixedP} and the matrix $DG^\frac{1}{2}W >0$, where $G^\frac{1}{2}$ is the positive definite square root of the gram matrix  $G$ of the $\tket{\psi}{ij_i}$ vectors and $D$ is defined in equation \eqref{DX}. Using $D$ we construct a new set of vectors $\{ \tket{\chi}{ij_i} \}_{i=1, j_i = 1}^{i=m, j_i = r_i}$, as given by equation \eqref{chi} and from this set we constuct a new ensemble of states $\{q_i, \sigma_i \}_{i=1}^{m}$, using equations \eqref{sigma} and \eqref{qi}. It was verified that the optimal POVM $\{ \Pi_i \}_{i=1}^{m}
$ is the PGM of the ensemble  $\{q_i, \sigma_i \}_{i=1}^{m}$. 

We now make the $U \left( r_1 \right) \times U \left( r_2 \right)  \times \cdots \times U \left( r_m \right)$ degree of freedom in choosing the pure state decomposition in equation \eqref{rhodecomposition} explicit. 

Thus, let $p_i \rho_i = \sum_{j=1}^{r_i} \ketbrat{\psi'}{ij}{\psi'}{ij}$ be a pure state decomposition of the LI states in the ensemble $p_i \rho_i, \; \forall \; 1 \leq i \leq m$, where $\tket{\psi'}{ij_i}$ and $\tket{\psi}{ij_i}$ are related by equation \eqref{psi'}, where $U'_D$ is a block diagonal unitary given by equation \eqref{U'D}. $U'_D$ is a variable for now; it's value will be fixed later. Corresponding to the primed vectors $\tket{\psi'}{ij_i}$, we have $\tket{u}{ij_i} \longrightarrow \tket{u'}{ij_i}$, as per equation \eqref{u'}, $W \longrightarrow W' = {U'_D}^\dag W U'_D$, $G \longrightarrow G'= {U'_D}^\dag G U'_D$ , $G^\frac{1}{2}W \longrightarrow  {G'}^\frac{1}{2}W' = {U'_D}^\dag G^\frac{1}{2}W U'_D$ and $\ket{w_{ij_i}} \longrightarrow \ket{w'_{ij_i}} = \sum_{l=1}^{m}\sum_{k_l=1}^{r_l} \left( {G'}^\frac{1}{2} W' \right)^{(l \; i)}_{k_l \; j_i} \tket{u'}{l k_l}$ (equation \eqref{w'expandu'}). $G^\frac{1}{2}W \longrightarrow {U'_D}^\dag G^\frac{1}{2}W U'_D$ implies that $X^{ \left( i j \right)} \
longrightarrow  {X'}^{ \left( 
i j \right)}= {{U'}^{ \left( i \right)}}^\dag X^{ \left( i j \right)} {U'}^{ \left( j \right)}$. In particular we can choose ${U'}^{ \left( i \right)}$ to be such that $ {X'}^{ \left( i i \right)}$ are diagonal, $\forall \; 1 \leq i \leq m$. This fixes the block diagonal unitary $U'_D$. Since $D \longrightarrow D' = {U'_D}^\dag D U'_D$, $D'$ is a diagonal matrix. This implies that $\tket{\chi}{ij_i}\longrightarrow \tket{\chi'}{ij_i} = \sum_{l=1}^{m}\sum_{k_l=1}^{r_l} \left(D'\right)^{(l \; i)}_{k_l \; j_i} \tket{\psi'}{l k_l}$ $=\left(D'\right)^{(i \; i)}_{j_i \; j_i} \tket{\psi'}{ij_i} $. As noted in the end of subsection \eqref{PQcorr}, the ensemble $\{q_i, \sigma_i \}_{i=1}^{m}$ remains invariant. Note that the diagonal of $D' {G'}^\frac{1}{2} W' = \sqrt{D' G' D'}$ is ${D'}^{2}$.

Let $\tket{\psi'}{ij_i} = \sqrt{\lambda_{ij_i}}\ket{\psi_{ij_i}}$, where $\ket{\psi_{ij_i}}$ are normalized. According to \cite{Singal} to solve the MED of the LI pure state ensemble $\{\lambda_{ij_i}, \ketbra{\psi_{ij_i}}{\psi_{ij_i}} \}_{i=1, j_i = 1}^{i=m, j_i = r_i}$, we need to find an $n \times n$ positive definite diagonal matrix $D''$, such that the diagonal of the positive square root of the matrix $D'' G' D''$ is ${D''}^2$. Here $G'$ is the gram matrix corresponding to the ensemble $\{\lambda_{ij_i}, \ketbra{\psi_{ij_i}}{\psi_{ij_i}} \}_{i=1, j_i = 1}^{i=m, j_i = r_i}$. But we have already found the solution: $D'' = D'$. In this case the optimal POVM is then given by $\{ \ketbra{w'_{ij_i}}{w'_{ij_i}} \}_{i=1,j_i=1}^{i=m, j_i=r_i}$. And we know that $\Pi_i = \sum_{j=1}^{r_i} \ketbra{w'_{ij_i}}{w'_{ij_i}}$, $\forall \; 1 \leq i \leq m$.
Also, just as shown in \cite{Singal}, $\{ \ketbra{w'_{ij_i}}{w'_{ij_i}} \}_{i=1,j_i=1}^{i=m, j_i=r_i}$ is the PGM of the ensemble $\{ \lambda'_{ij_i}, \ketbra{\psi_{ij_i}}{\psi_{ij_i}} \}_{i=1,j_i = 1}^{i=m,j_i=r_i}$, where $\lambda'_{ij_i} =\dfrac{ \left(   \left( D' \right)^{(i \; i)}_{j_i j_i} \right)^2 \lambda_{ij_i}  } { Tr\left(D' G' D'\right)}$. But just as noted above $\tket{\chi'}{ij_i} = \sqrt{ Tr \left( D'G'D' \right) } \sqrt{\lambda'_{ij_i}}\ket{\psi_{ij_i}}$, $\forall \; 1 \leq i \leq m, \; 1 \leq j_i \leq r_i$. \emph{Thus, $\{ \lambda'_{ij_i}, \ketbra{\psi_{ij_i}}{\psi_{ij_i}} \}_{i=1,j_i = 1}^{i=m,j_i=r_i}$, whose PGM is the optimal POVM for the ensemble $\{ \lambda_{ij_i}, \ketbra{\psi_{ij_i}}{\psi_{ij_i}} \}_{i=1,j_i = 1}^{i=m,j_i=r_i}$, is a pure state decomposition of the ensemble $\{q_i, \sigma_i \}_{i=1}^{m}$ $ ( =\mathscr{R} \left( \widetilde{P} \right))$, whose PGM is the optimal POVM for the ensemble $\{ p_i, \rho_i \}_{i=1}^{m}$, where $\{ \lambda_{ij_i}, \ketbra{\psi_{ij_i}}{\psi_{
ij_i}} \}_{i=1,j_i = 1}^{i=m,j_i=r_i}$ itself is a pure state decomposition of 
the ensemble $\{ p_i, \rho_i \}_{i=1}^{m}$}.

The feature that ensures that there is a LI pure state decomposition of the mixed state ensemble, such that the optimal POVM of the LI pure state ensemble is a pure state decomposition of the optimal POVM of the LI mixed state ensemble, is the spectral decomposition of the matrices $X^{(ii)}$. This begs the question: for any LI mixed state ensemble, is such a LI pure state decomposition unique? The key feature that is required is that the ${X'}^{(ii)}$ matrices are diagonalized. Hence there are as many pure state decompositions of the mixed state ensemble with this property as there are spectral decompositions of the $D$ matrix. If $X^{(ii)}$ has $s_i$ distinct eigenvalues and the degeneracy of the $j_{i}$-th eigenvalue $( 1 \leq j_i \leq s_i)$ has a degeneracy of $k_{j_i}$\footnote{Needless to say, $\sum_{j_i}^{s_i} k_{j_i} = r_i$.}, then there is a $U(k_{1_1}) \times U(k_{2_1})  \times \cdots \times U(k_{s_1}) \times U(k_{1_2}) \times U(k_{2_2})  \times \cdots \times U(k_{s_2}) \times \cdots \times U(k_{1_
m}) \times U(k_{2_m})  \times \cdots \times U(k_{s_m})$ degree of freedom in choosing a pure state decomposition of the mixed state ensemble with this property. 

What about the converse? Consider an ensemble of pure states $\{ \lambda_i, \ketbra{\psi_i}{\psi_i} \}_{i=1}^{n}$ whose optimal POVM is $\{ \ketbra{v_i}{v_i} \}_{i=1}^{n}$. Partition the ensemble into disjoint subsets and sum over the elements in each subset and collect all such summations to form a new ensemble $\{p_i, \rho_i \}_{i=1}^{m}$, whose optimal POVM, let's say is given by $\{\Pi_i \}_{i=1}^{m}$. It is generally not the case that $\{ \ketbra{v_i}{v_i} \}_{i=1}^{n}$ is a pure state decomposition of elements in $\{\Pi_i \}_{i=1}^{m}$. So for which pure state ensembles is this true? Let's re-index the pure state ensemble: $i \longrightarrow (i, j_i)$, so that $p_i \rho_i = \sum_{j=1}^{r_i} \lambda_i \ketbra{\psi_{ij_i}}{\psi_{ij_i}}$. While performing the MED of the LI pure state ensemble, if the matrix $DG^\frac{1}{2}W$ is such that its block diagonals\footnote{i.e., the first $r_1 \times r_1$ diagonal block, the second $r_2 \times r_2$ block etc} are diagonal, then it is easy to see that the 
relation $\Pi_i = \sum_{j=1}^{r_i} \ketbra{v_{ij}}{v_{ij}}$ also holds true.

Another question is if, given the problem of the MED of a LI mixed state ensemble, can one substitute the problem with the MED of a pure state decomposition such that the optimal POVM of the latter is a pure state decomposition of the former? The answer, unfortunately, is no. The reason being that to substitute the mixed state ensemble MED problem with the pure state ensemble MED problem one needs to first obtain the $n \times n$ unitary $W$ such that when $D G^\frac{1}{2} W$ is constructed (where $D$ is given by equation \eqref{DX}), it is positive definite. This is already equivalent to finding a solution for the MED of the mixed state ensemble. 

We know that the optimal POVM of a pure state LI ensemble is given by its own PGM iff the diagonal of the positive square root of the ensemble's gram matrix is a multiple of the identity. How does this condition change when we're given to perform the MED of a LI mixed state ensemble? In the following we prove that this occurs iff the diagonal blocks of  $G^\frac{1}{2}$ are diagonalized and when the diagonal of $G^\frac{1}{2}$ is a multiple of the identity. 

\begin{theorem}
\label{PGMOPT}
For an ensemble $\widetilde{P} \in \ens$ to satisfy $\mathscr{R}\left(\widetilde{P}\right) = \widetilde{P}$ it is necessary and sufficient that all eigenvalues of all the block diagonal matrices of $G^\frac{1}{2}$ are equal.
\end{theorem}
\begin{proof}
\textbf{Necessary Part:} Let $\mathscr{R}\left(\widetilde{P}\right) = \widetilde{P}$. Let the pure state decomposition of $\widetilde{P}$ whose optimal POVM is a pure state decomposition of the optimal POVM for MED of $\widetilde{P}$ be $\{ \lambda_{ij_i}, \ketbra{\psi_{ij_i}}{\psi_{ij_i}} \}_{i=1,j_i = 1}^{i=m,j_i=r_i}$. Hence we have $p_i \rho_i = \sum_{j=1}^{r_i} \lambda_{ij} \ketbra{\psi_{ij}}{\psi_{ij}}$, $\forall \; 1 \leq i \leq m$. It was mentioned above that there exists a pure state decomposition of $\mathscr{R}\left( \widetilde{P} \right)$ of the form $\{ \lambda'_{ij_i}, \ketbra{\psi_{ij_i}}{\psi_{ij_i}} \}_{i=1,j_i = 1}^{i=m,j_i=r_i}$, who PGM is the optimal POVM of $\{ \lambda_{ij_i}, \ketbra{\psi_{ij_i}}{\psi_{ij_i}} \}_{i=1,j_i = 1}^{i=m,j_i=r_i}$. Since $\mathscr{R}\left(\widetilde{P}\right) = \widetilde{P}$ it follows that the $\sqrt{\lambda'_{ij}}\ket{\psi_{ij}}$ ($ = \tket{\psi'}{ij_i}$) vectors and the $\sqrt{\lambda_{ij}}\ket{\psi_{ij}}$ ($ = \tket{\psi}{ij_i}$) vectors are related by a 
block diagonal unitary transformation, given in equation \eqref{psi'}. But since the set $\{ \ket{\psi_{ij}} \}_{j=1}^{r_i}$ are linearly independent, it follows that $U'_D$ must be a diagonal matrix. This can only mean that both the ensembles $\{ \lambda'_{ij_i}, \ketbra{\psi_{ij_i}}{\psi_{ij_i}} \}_{i=1,j_i = 1}^{i=m,j_i=r_i}$ and $\{ \lambda_{ij_i}, \ketbra{\psi_{ij_i}}{\psi_{ij_i}} \}_{i=1,j_i = 1}^{i=m,j_i=r_i}$ are equal, as well. In the beginning of section \eqref{compareMEDP}, it was noted that $\sqrt{\lambda'_{ij_i}} \ket{\psi_{ij_i}}$ and $\sqrt{\lambda_{ij_i}} \ket{\psi_{ij_i}}$ are also related through $\lambda'_{ij_i} =\dfrac{ \left(   \left( D' \right)^{(i \; i)}_{j_i j_i} \right)^2 \lambda_{ij_i}  } { Tr\left(D' G' D'\right)}$, $\forall \; 1 \leq i \leq m, \; 1 \leq j_i \leq r_i$. But since $\lambda'_{ij_i} = \lambda_{ij}$ , $\forall \; 1 \leq i \leq m, \; 1 \leq j_i \leq r_i$, this implies that $D'$ is a multiple of the identity. This 
implies that $D' G' D' \propto G'$ which implies that $ \sqrt{D' G' D'} = D' {G'}^\frac{1}{2} W' \propto {G'}^\frac{1}{2}$. This implies that $W = \mathbb{1}_n$. $D'$ is the block diagonal matrix one gets by ``extracting'' the diagonal blocks of ${G'}^\frac{1}{2} W'$. Since $W' = \mathbb{1}_n$, $D'$ is the block diagonal matrix ``extracted'' from ${G'}^\frac{1}{2}$. Similarly, $D$ is the block diagonal matrix ``extracted'' from ${G}^\frac{1}{2}$. And since $D'$ is a multiple of identity and since $D'$ and $D$ are related by a unitary transformation, $D$ is also a multiple of the identity. This tells us that the diagonal blocks of $G^\frac{1}{2}$, i.e., the matrices $X^{(ii)}$, are positive definite diagonal matrices, with equal diagonals. Hence all eigenvalues of all the block diagonal matrices of $G^\frac{1}{2}$ are equal.

\textbf{Sufficient Part:} If all eigenvalues of all the block diagonal matrices of $G^\frac{1}{2}$ are equal, then these diagonal blocks are diagonal matrices themselves. Let $D''$ be the matrix comprising of only the diagonal blocks of $G^\frac{1}{2}$. $D''$ is, thus, a multiple of the identity. Note that the the block-diagonal of $D'' {G'}^\frac{1}{2}$ is $ {D''}^2$. Hence we have found a block-diagonal positive definite matrix $D''$ such that the diagonal blocks of the positive square root of $D'' G D''$ is given by ${D''}^2$, which implies that we have solved the MED problem for the ensemble $\widetilde{P}$. Using $D''$, we can construct the vectors $\tket{\chi}{ij_i}$ from the vectors $\tket{\psi}{ij_i}$ using equation \eqref{chi} and then construct the states $q_i \sigma_i$ using equation \eqref{sigma} and equation \eqref{qi}. Since $D''$ is simply a multiple of the identity, it isn't difficult to see that $q_i \sigma_i = p_i \rho_i$, $\forall \; 1 \leq i \leq m$. This proves that $\mathscr{R}\left(R   
\right) = P$.

Hence proved. 
\end{proof}

\section{Solution For the MED problem}
\label{Solution}

The necessary and sufficient condition to solve the MED for a general LI ensemble as specified by \textbf{A} (on page \pageref{AA}) suggest a technique to solve the problem. In this section we give this technique without going into the theoretical details which justify the claim that it can be used effectively to obtain the optimal POVM for the MED of any ensemble $\widetilde{P} \in \ens$. This is because this techhnique is a generalization of the technique given in \cite{Singal}, wherein all the relevant theoretical background has been developed for LI of pure state ensembles. The theoretical background for the mixed states ensemble case is a trivial generalization of that for the pure state ensemble case; it follows the same sequence of steps as that for the LI pure state ensemble case. In the following we explain what the technique is. 

We assume that we know the solution for the MED of some ensemble $\widetilde{P}_0=\{ p^{(0)}_i, \rho^{(0)}_i \}_{i=1}^{m} \in \ens$ and want to obtain the solution for the MED of another ensemble $\widetilde{P}_1=\{ p^{(1)}_i, \rho^{(1)}_i  \}_{i=1}^{m} \in \ens$. Let $p^{(0)}_i \rho^{(0)}_{i} = \sum_{j=1}^{r_i} \ketbrat{\psi^{(0)}}{ij}{\psi^{(0)}}{ij}, \; \forall \; 1 \leq i \leq m,$ be a pure state decomposition for the ensemble $\widetilde{P}_0$. And let the gram matix corresponding to the set $\{ \tket{\psi^{(0)}}{ij_i} \}_{i=1,j_i =1 }^{m,r_i}$ be $G_0$. Similarly, let $p^{(1)}_i \rho^{(1)}_{i} = \sum_{j=1}^{r_i} \ketbrat{\psi^{(1)}}{ij}{\psi^{(1)}}{ij}, \; \forall \; 1 \leq i \leq m,$ be a pure state decomposition for the ensemble $\widetilde{P}_1$. And let the gram matix corresponding to the set $\{ \tket{\psi^{(1)}}{ij_i} \}_{i=1,j_i =1 }^{m,r_i}$ be $G_1$. Knowing the solution for the MED of $\widetilde{P}_0$ implies that we know a block diagonal matrix $D_0$, of the form as given by equation \eqref{DX}, such that the diagonal-block of positive square root of $D_0G_0D_0$ is $D_0^2$ (in accordance with the rotationally invariant necessary and sufficient conditions given by \textbf{A} on page \pageref{AA}). Let's rewrite equation \eqref{Ainv} in the following form:

\begin{equation}
\label{EOY}
\left( D G^\frac{1}{2} W \right)^2 - DGD = 0
\end{equation}

Let's define a linear function $G(t) \equiv (1-t) G_0 + t G_1$, where $t \in [0,1]$. So $G(0)=G_0$ and $G(1)=G_1$. Note that $G(t) > 0$ and $Tr \left( G(t) \right) =1, \; \forall \; 0 \leq t \leq 1$. Thus for any value of $t \in [0,1]$, $G(t)$ corresponds to the gram matrix of a pure state decomposition of some ensemble $\widetilde{P}_t \in \ens$\footnote{Actually, $G(t)$, for each value of $t \in [0,1]$, corresponds to a family of unitarily equivalent ensembles, i.e., $G(t)$ corresponds to the set of ensembles $\{ U \widetilde{P}_t U^\dag \; | \; U \text{ varies over } U(n) \}$. The notation $U \widetilde{P}_t U^\dag$ is the same as has been used in equation \eqref{ens1}.}. Using equation \eqref{EOY} we drag the solution for $D$ from $t=0$ where the value is known to $t=1$ where the solution isn't known. This can be done in different ways. 

\subsection{Taylor Series Expansion and Analytic Continuation} 
\label{Taylor}

A formal way of doing it is by using Taylor series expansion and analytic continuation. We start by assuming that the matrices $\left( DG^\frac{1}{2}W \right) (t), D(t)$ and $G(t)$ are analytic functions from $[0,1]$\footnote{$G(t)$ is the function mentioned above; it is linear in $t$ and hence is analytic in $t$. We will not provide for the proof of the analytic dependence of $\left( DG^\frac{1}{2}W \right) (t)$ or $ D(t)$ here since a detailed proof the same is provided in \cite{Singal} for the pure state ensemble case which can be trivially generalized to the mixed state case.}. $D(t)$ take the form

\begin{equation}
\label{DT}
D(t) = \begin{pmatrix}
            {Z^{(11)(t)}} & 0 &  \cdots & 0 \\
            0 & {Z^{(22)(t)}} &  \cdots & 0\\
            \vdots & \vdots & \ddots & \vdots \\
            0 & 0 & \cdots & {Z^{(mm)(t)}}
           \end{pmatrix},
\end{equation}

and $\left( DG^\frac{1}{2} W \right) (t)$ takes the form

\begin{equation}
\label{DGWT}
\left( DG^{\frac{1}{2}} W \right) (t)   = \begin{pmatrix}
                       \left( Z^{(11)} (t) \right)^2 &  Z^{(12)}(t) & \cdots & Z^{(1m)}(t)\\
                       Z^{(21)}(t) & \left( Z^{(22)} (t) \right)^2 & \cdots & Z^{(2m)} (t)\\
                       \vdots & \vdots & \ddots & \vdots\\
                       Z^{(m1)}(t) & Z^{(m2)}(t) & \cdots & \left( Z^{(mm)}(t) \right)^2
                       \end{pmatrix},
\end{equation}

where

\begin{equation}
\label{ZijT}
Z^{(ij)}(t) = \begin{pmatrix}
            r^{(ij)}_{11}(t) + i c^{(ij)}_{11}(t) & r^{(ij)}_{12}(t) + i c^{(ij)}_{12}(t) &  \cdots & r^{(ij)}_{1r_j}(t) + i c^{(ij)}_{1r_j}(t) \\
            r^{(ij)}_{21}(t) + i c^{(ij)}_{21}(t) & r^{(ij)}_{22}(t) + i c^{(ij)}_{22}(t) &  \cdots & r^{(ij)}_{2r_j}(t) + i c^{(ij)}_{2r_j}(t) \\
            \vdots & \vdots & \ddots & \vdots \\
            r^{(ij)}_{r_i 1}(t) + i c^{(ij)}_{r_i1}(t) & r^{(ij)}_{r_i2}(t) + i c^{(ij)}_{r_i2}(t) &  \cdots & r^{(ij)}_{r_ir_j}(t) + i c^{(ij)}_{r_ir_j}(t) \\
           \end{pmatrix},
\end{equation}
i.e, $Z^{(ij)}(t)$ are $r_i \times r_j$ matrices. Also, the hermiticity of $\left( DG^\frac{1}{2} W \right) (t)$ requires that  $c^{(ii)}_{jk}(t) = - c^{(ii)}_{kj}(t)  , \; 1 \leq i \leq m, \; 1 \leq j, k \leq r_i$. With the constraints on $c^{(ii)}_{jk}(t)$ in place, $r^{(il)}_{j_ik_l}$ and $c^{(il)}_{j_ik_l}$ are $n^2$ (dependent) variables. In the following we show how to obtain the Taylor series expansion of these variables with respect to the independent variable $t$. Note that $Z^{(ii)}$ are equal to $X^{(ii)}$ and $Z^{(ij)}(t)$ are equal to $ \left( X^({ii}) \right)^{-1} X^{(ij)}$ forall $1 \leq i \neq j \leq m$, where $X^{(ij)}$ are defined in equation \eqref{part1}. 

Taking the total derivative of both sides of equation \eqref{EOY} with respect to $t$ and set $t= 0$, we get $n^2$ coupled linear equations which can be solved for the unknowns $\dfrac{d r^{(i l)}_{j_i k_l}}{dt} |_{t=0}$ and $\dfrac{d c^{(i l)}_{j_i k_l}}{dt} |_{t=0}$, $\forall \; 1 \leq i, l \leq m,$ $1 \leq j_i \leq r_i$ and $1 \leq k_l \leq r_l$. Again taking the second order total derivative of both sides of equation \eqref{EOY} with respect to $t$ and setting $t=0$, we get $n^2$ coupled linear equations which can be solved for the unknowns $\dfrac{d^2 r^{(i l)}_{j_i k_l}}{dt^2} |_{t=0}$ and $\dfrac{d^2 c^{(i l)}_{j_i k_l}}{dt^2} |_{t=0}$, $\forall \; 1 \leq i, l \leq m,$ $1 \leq j_i \leq r_i$ and $1 \leq k_l \leq r_l$. In this way we have obtain the $K$-th order derivatives of the $r^{il}_{j_ik_l}$ and $c^{il}_{j_ik_l}$ with respect to $t$ at $t=0$. Using these derivatives we can taylor expand about $r^{il}_{j_ik_l}(t)$ and $c^{il}_{j_ik_l}(t)$ about $t=0$. Our goal is to find a solution for the values 
of $r^{il}_{j_ik_l}(1)$ and $c^{il}_{j_ik_l}(1)$. It is reasonable to divide the interval $[0,1]$ into a certain number of intervales, say $L$ intervals, so that one taylor expands within every interval and then analytically continues from the starting point of each interval to reach $t=1$ finally. The following statements are made on the basis of results in \cite{Singal}:

\begin{itemize}
\item $L\equiv \lceil|| G(0) - G(1) || n^{2}\rceil$ gives a reasonable number of intervals for very low error margin. Also beyond a certain order to which Taylor series are expanded the error margin doesn't decrease appreciably; neither does the error margin increase appreciably as $n$ increases while the order to which Taylor series is expanded remained constant. 
\item For ensembles $\widetilde{P}_1 \in \ens$, to which the gram matrix $G_1$ corresponds, one can find the starting point gram matrix $G_0$ close enough to $G_1$ such that $\lceil|| G(0) - G(1) || n^{2}\rceil=1$. This implies that the interval $[0,1]$ need not be divded into subintervals for the purpose of analytic continuation. For these cases the computational complexity of such process is $n^6$. In cases where one isn't able to obtain the starting point close enough, the computional complexity increases to $n^8$, as expected. This is because the number of intervals required to obtain the solution increases as $n^2$ with $n$. 
\end{itemize}

\subsection{Newton-Raphson Method}
\label{Newton}

Another technique to obtain the solution for the the MED of optimal POVM for a LI ensemble is to use Newton's method based on equation \eqref{EOY}. As starting point, we substitute the solutions for the MED of $G_0$ viz., $D_0$ and $D_0G_0^\frac{1}{2} W_0$, whose values we know, in equation \eqref{EOY}, along with $G_1$. The aim is to change the values of $D$ and $DG^\frac{1}{2}W$ so that the LHS of the equation converges to $0$\footnote{Despite the fact that we have no formal proof that Newton-Raphson method will necessarily converge to the desired solution for equation \eqref{EOY}, over 100,000 examples for various values of $n$ and $r_1, r_2, \cdots, r_m$ have been sampled, for which the method works. An undesirable solution would require that the LHS of equation \eqref{EOY} does converge to $0$ but that $DG^\frac{1}{2}W$ isn't positive definite. Heuristically, we can expect $D$ and $D G^\frac{1}{2} W$ to converge to the desirable solution (i.e., the solution such that $D_1G_1^\frac{1}{2}W_1 >0$) because 
our starting point has that $D_0 G_0^\frac{1}{2} W_0>0$ and is, hence, likely to be ``closer'' to our starting point; the metric being given by the Hilbert-Schmidt norm.}. The sequence of steps are the same as laid out in \cite{Singal}. This method is much simpler to implement compared to the Taylor series example and has a computational complexity of $n^6$.

\subsection{Barrier Type Interior Point Method (SDP)}
\label{SDP}

In \cite{Singal} we showed how the barrier-type interior point method has a computational complexity of $n^8$. We will summarize in brief how this barrier-type interior point method works. This is an iterative algorithm just like the Newton-Raphson method. In fact, the barrier type interior point method comprises of implementing the Newton-Raphson method to obtain the stationary point of the quantity being minimized which is the the quantity $Tr \left( Z \right) - \sum_{i=1}^{m}w_i^{(0)} Log \left( Det \left( Z - p_i \rho_i \right) \right)$. The weights $w_i^{(0)}$ have a very small value, so that the objective function varies by very little from the function $Tr \left( Z \right)$. The reason the term $\sum_{i=1}^{m}w_i^{(0)} Log \left( Det \left( Z - p_i \rho_i \right) \right)$ is added to the function $Tr \left( Z \right) $ is to ensure that if we start from a feasible point (a point where $Z^{(0)} - p_i \rho_i \geq 0, \; \forall \; 1 \leq i \leq m$), our second iterate $Z^{(1)}$ will necessarily remain in 
the feasible region. This happens because the term $\sum_{i=1}^{m}w_i^{(0)} Log \left( Det \left( Z - p_i \rho_i \right) \right)$ blows up to infinity if any of the operators $ Z - p_i \rho_i$ approaches the boundary of the positive convex set, i.e., if the eigenvalue(s) of any one of these operators approaches 0; the directional derivative would be such that the next iterate for $Z$ would remain in the feasible region. Computing the directional derivative involes computing an $n^2 \times n^2$ square matrix whose computational cost is $n^8$. Thus the computational cost of the barrier-type interior point method is $n^8$. 

\textbf{ The Taylor series method and Newton-Raphson method mentioned in sections \eqref{Taylor} and \eqref{Newton} have lower computational complexity and are simpler to implement thus giving an edge over the SDP method mentioned above. }

\section{Conclusion}

We look back over what has been done in this paper: first, the necessary and sufficient conditions for obtaining the optimal POVM for the MED for an ensembles of linearly independent states was simplified. Using the simplified conditions we proved that there exists a bijective function $\mathscr{R}$ which when acted upon any such ensemble gives another ensemble whose PGM is the optimal POVM of for the MED of the pre-image. We also obtained a closed form expression for $\mathscr{R}^{-1}$. This is a generalization of a similar result that was hitherto only proved for linearly independent pure state ensemble in \cite{Bela, Mas, Carlos}. The result also gives a rotationally invariant form of representing the necessary and sufficient conditions for the MED of an ensemble of LI states. This rotationally invariant form for the necessary and sufficient conditions of the optimal POVM is employed for two purposes: 1.) we use it to show that for every LI mixed state ensemble there exists a corresponding pure state 
decomposition so that the optimal POVM for the MED of the latter is a pure state decomposition for the MED of the former. This is then employed to show under what conditions the optimal POVM of a mixed state ensemble is given by its own PGM. 2.) We employ this rotationally invariant form of the necessary and sufficient conditions in a technique which gives us the optimal POVM for an ensemble. Our technique is compared to a standard SDP technique; that of a barrier-type interior point method. It is found that along with the advantage of our technique being simpler to implement, our technique has a lower computational complexity compared to the barrier-type IPM; our technique has a computational complexity of $n^6$ whereas the computational complexity of the latter SDP technique is $n^8$, which gives our technique an edge over the SDP technique.

%\end{multicols}


\begin{thebibliography}{1}
\bibitem{gram} The gram matrices we deal with in this paper always have trace one.
\end{thebibliography}
\end{document}